\documentclass[12pt]{article}
\usepackage{fullpage}

\usepackage{amsmath}
\usepackage{amssymb, amsfonts}
\usepackage{amsthm}
\theoremstyle{plain}
\newtheorem{theorem}{Theorem}
\newtheorem{lemma}{Lemma}
\newtheorem{proposition}{Proposition}

\theoremstyle{definition}
\newtheorem{assumption}{Assumption}

\theoremstyle{remark}
\newtheorem{remark}{Remark}

\usepackage{booktabs}
\usepackage{graphicx}
\usepackage[colorlinks=true, allcolors=blue]{hyperref}
\usepackage[round, sort, authoryear]{natbib}

\def\T{{\!\top\!}}
\def\tr{\mathrm{tr}}
\def\diag{\mathrm{diag}}
\newcommand{\ind}[1]{\mathbb{I}(#1)}
\def\pr{\mathrm{P}}
\def\E{\mathrm{E}}
\def\var{\mathrm{Var}}
\def\cov{\mathrm{Cov}}
\usepackage{bm}
\usepackage{mathrsfs}

\begin{document}
\title{\textbf{Computationally efficient and data-adaptive changepoint inference in high dimension}}
\author{Guanghui Wang\\ 
	East China Normal University
	\and Long Feng\thanks{Corresponding author: flnankai@nankai.edu.cn}\\ 
	Nankai University}
\date{\today}
\maketitle
\baselineskip 20pt

\begin{abstract}
High-dimensional changepoint inference that adapts to various change patterns has received much attention recently. We propose a simple, fast yet effective approach for adaptive changepoint testing. The key observation is that two statistics based on aggregating cumulative sum statistics over all dimensions and possible changepoints by taking their maximum and summation, respectively, are asymptotically independent under some mild conditions. Hence we are able to form a new test by combining the p-values of the maximum- and summation-type statistics according to their limit null distributions.
To this end, we develop new tools and techniques to establish asymptotic distribution of the maximum-type statistic under a more relaxed condition on componentwise correlations among all variables than that in existing literature.
The proposed method is simple to use and computationally efficient.
It is adaptive to different sparsity levels of change signals, and is comparable to or even outperforms existing approaches as revealed by our numerical studies.
\end{abstract}

\vspace{0.2cm}

\section{Introduction}\label{sec:introduction}

Heterogeneity is a ubiquitous feature of high dimensional data. A particular form of heterogeneity is the changepoint structure,
that is, the data generation mechanism may occur a sudden change at some time point or location.
The primary aims of changepoint detection are (i) testing for the existence of any changepoint and (ii) estimation of one or a few more changepoints if they exist. We refer readers to, for example, \cite{10.1111/j.1467-9892.2012.00819.x} and \cite{10.1214/16-sts587} for an overview.
In this paper, we focus on the testing aspect in high dimension.
Consider a sequence of $p$-dimensional vectors of sample size $n$, i.e., $\{X_i=(X_{i1},\ldots,X_{ip})^\T\}_{i=1}^n$, from the following mean-change model
\begin{align*}
    X_i = \mu_0 + \delta \ind{i>\tau} + \epsilon_i,\ i=1,\ldots,n,
\end{align*}
where $\mu_0\in\mathbb{R}^p$ represents the baseline mean level, $\delta\in\mathbb{R}^p$ is the mean-change signal parameter,
$\tau\in\{1,\ldots,n\}$ is the possible changepoint, and $\{\epsilon_i=(\epsilon_{i1},\ldots,\epsilon_{ip})^\T\}_{i=1}^n$ are random noises with zero mean.
Of interest is to test whether there exists a changepoint, that is,
\begin{align}\label{H0}
\begin{gathered}
    H_0: \tau=n\ \text{and}\ \delta=0\ \text{versus}\\
    H_1:\ \text{there exists}\ \tau\in\{1,\ldots,n-1\}\ \text{and}\ \delta\neq 0,
\end{gathered}
\end{align}
under the paradigm that both the sample size $n$ and dimension $p$ grow to infinity.

A review of recent developments of various testing procedures for \eqref{H0} can be found in \cite{10.1016/j.jmva.2021.104833}.
Most prominent are based on a sequence of individual cumulative sum (CUSUM) statistics, to wit, $\{C_{\gamma, j}(k),\ j=1,\ldots,p\}_{k=1}^{n-1}$ with $\gamma=0$ or $0.5$ frequently used, where
\begin{align}\label{CUSUM}
    C_{\gamma, j}(k) = \left\{\frac{k}{n}\left(1-\frac{k}{n}\right)\right\}^{-\gamma}\frac{1}{\sqrt{n}}\left(S_{kj}-\frac{k}{n}S_{nj}\right)/\widehat{\sigma}_{j},
\end{align}
$S_{kj}=\sum_{i=1}^{k}X_{ij}$ and $\widehat{\sigma}_{j}$'s are estimators for the (long-run) variances.
By first aggregating the individual CUSUMs over $p$ dimensions at every possible changepoint $k\in\{1,\ldots,n-1\}$, a test statistic can then be constructed based on the maximum or summation of all aggregations over all possible $k$'s.
The difference between $\gamma=0$ or 0.5 lies in the power performance. To be specific, the power with $\gamma=0$ may decay if the changepoint appears early or late \citep{csorgo1997limit}.

Different ways of aggregations are discussed in the literature. Among them, \cite{10.1016/j.jeconom.2009.10.020}, \cite{10.1111/j.1467-9892.2012.00796.x} and \cite{10.1007/s11425-016-0058-5} considered $L_2$-aggregations followed by the maximum operator, i.e., $\max_{k=1,\ldots,n-1}\sum_{j=1}^p C_{0, j}^2(k)$.
A strategy that replaces each $C_{0, j}^2(k)$ by a self-normalized U-statistic was suggested in \cite{wang2019inference}.
With proper normalization, such max-$L_2$-type statistic converges to the supremum of some function of a Gaussian process under necessary conditions if $H_0$ holds, for which the corresponding quantile is usually obtained via simulations.
Instead of applying the maximum operator to all $L_2$-aggregations, \cite{10.1214/17-aos1610} and \cite{10.1142/s201032631950014x} proposed a sum-$L_2$-type statistic, namely, $\sum_{k=1}^{n-1}\sum_{j=1}^p C_{0.5, j}^2(k)$,
which, after appropriate normalization, is shown to asymptotically admit a Gaussian distribution under $H_0$.
The $L_\infty$-aggregations in conjunction with the maximum operator has also attracted much attention in the literature.
For instance, \cite{10.1214/15-aos1347} advocated the statistic $\max_{k=1,\ldots,n-1}\max_{j=1,\ldots,p}|C_{0, j}(k)|$ and showed that it converges in distribution to the extreme value distribution of Gumbel type after suitable normalization under $H_0$.
In addition, \cite{10.1111/rssb.12406} considered the statistic $\max_{\lambda\leq k\leq n-\lambda}\max_{j=1,\ldots,p}|C_{0.5, j}(k)|$ with $\lambda\in[1,n/2]$ being a user-specified boundary removal parameter, and investigated how to approximate its null distribution via the multiplier bootstrap; they didn't find the asymptotic null distribution of the test statistic.

It is well-known that the $L_2$-aggregation method is more effective in detecting dense and weak change signal in the sense that a large number of entries in $\delta$ are non-zero each with a small magnitude, and the $L_\infty$-aggregation performs better for sparse but strong change signal when there exist a few number of non-zero entries in $\delta$ with large magnitudes.
The impact of sparsity on the detection boundary was investigated by \cite{10.1214/18-aos1740} and \cite{10.1214/20-aos1994} for Gaussian data sequences.
They also provided testing rules that can achieve the minimax detection rate and are adaptive to the sparsity by taking supremum over a grid of sparsity levels in conjunction with well-calibrated thresholds.
Recently, there is a great deal of effort to develop adaptive testing procedures that are effective for various alternative change patterns.
For instance, \cite{10.1111/rssb.12079} and \cite{10.1214/16-ejs1155} proposed only aggregating those CUSUMs at every possible changepoint $k$ (i.e., $\{C_{0.5, j}(k),\ j=1,\ldots,p\}$) that pass a certain threshold, and they showed that such testing rules have a vanishing Type I error.
More generally, \cite{10.1111/rssb.12375} first introduced an adjusted $L_q$-aggregation strategy, to wit, with $1\leq q\leq\infty$ and $1\leq s_0\leq p$, $T_{q,s_0} := \max_{\lambda\leq k\leq n-\lambda}\big\{\sum_{j=1}^{s_0}|C_{0, (j)}(k)|^q\big\}^{1/q}$,
where $|C_{0, (1)}(k)|\geq\cdots\geq|C_{0, (p)}(k)|$ is the order statistics of $\{|C_{0, j}(k)|\}_{j=1}^p$.
They then proposed an adaptive procedure by taking the minimum of p-values associated with every $T_{q,s_0}$ among a range of choices of $q$ (e.g., $q\in\{1,2,3,4,5,\infty\}$) with a fixed $s_0$, where each p-value was approximated by the multiplier bootstrap.
Combining different $L_q$-aggregations (without adjustment) was considered by \cite{doi:10.1080/01621459.2021.1884562}. Specifically, for an even $q$,
they first proposed a sequence of self-normalized U-statistics at every possible changepoint to estimate $\|\delta\|_q^q$ scaled by the location of possible changepoint, say $\{U_q(k)\}_{k=2q}^{n-2q}$.
Then the asymptotic null distribution of $T_q := \max_{k=2q,\ldots,n-2q} U^2_q(k)$ is derived and is provably pivotal.
At last, they used the minimum p-value combination over different values of $q$ (e.g., $q\in\{2,6\}$) to construct an adaptive test based on an important fact that $T_q$'s with different even $q$'s are asymptotically independent under $H_0$.
Projection is an alternative method for constructing adaptive testing rules.
\cite{10.1214/18-ejs1442} proposed projecting the CUSUMs $(C_{0,1}(k),\ldots,C_{0,p}(k))^\T$ for each $k=1,\ldots,n-1$ along a random direction and proved that the projections, after suitable normalization, converge to a standard Brownian bridge under $H_0$; thus the maximum or summation operator can be applied over all $k$'s to conduct a test.
\cite{10.1111/rssb.12243} provided a data-driven estimator of the direction along which the CUSUMs $(C_{0.5,1}(k),\ldots,C_{0.5,p}(k))^\T$ should be projected,
and showed that the test has a vanishing Type I error.

In this paper, we suggest a new adaptive method that is simple, fast yet effective.
The key observation is that the max-$L_\infty$- and sum-$L_2$-type statistics are asymptotic independent under the null hypothesis that there is no changepoint.
Hence a test can be conducted by combining the two separate p-values.
We refer to the proposed test as double-max-sum (DMS) method, since we are essentially conducting the maximum or summation operator along both dimensions and possible changepoints.
The contributions of this paper is three-fold.
\begin{itemize}
    \item[(a)] The DMS is computationally efficient. The test statistic can be implemented in linear time over both $n$ and $p$. The critical value is determined according to the associated asymptotic null distribution, which is simple to use and thus avoids any numerical approximations such as by simulating Gaussian analogues or by using multiplier bootstrap.

    \item[(b)] We are able to show that each of two versions of max-$L_\infty$-type statistic, 
    to wit,
    \begin{align*}
        M_{n,p}:=\max_{k=1,\ldots,n-1}\max_{j=1,\ldots,p}|C_{0, j}(k)|\ \text{and}\
        M^{\dagger}_{n,p}:=\max_{\lambda_n\leq k\leq n-\lambda_n}\max_{j=1,\ldots,p}|C_{0.5, j}(k)|,
    \end{align*}
    with $\lambda_n\in[1,n/2]$ being a pre-specified boundary removal parameter,
    converges in distribution to the extreme value distribution of Gumbel type after appropriate normalization under some mild conditions if $H_0$ holds.
    The ways of convergence of $M_{n,p}$ and $M^{\dagger}_{n,p}$ differ in their normalizing factors.
    In the literature, \cite{10.1214/15-aos1347} studied the limit distribution of $M_{n,p}$ under a logarithmic decay assumption on componentwise correlations among $p$ variables, which generally requires a natural ordering among all variables and is hard to be verified in practice. With the presence of such ordering, he obtained the asymptotic null distribution by supplying the blocking arguments \citep{MR362465}. In this paper, we develop new tools and techniques based on the \textit{inclusion-exclusion principle} without restrictions on any ordering patterns.
    In addition, to our best knowledge, the study of the asymptotic distribution of $M^{\dagger}_{n,p}$ is new. In fact, it is related to the function $\max_{1\leq j\leq p}\sup_{\Lambda_p\leq t\leq 1-\Lambda_p}\{t(1-t)\}^{-1/2}|\mathcal{W}_{tj}-t\mathcal{W}_{tj}|$, with suitably chosen $\Lambda_p=o(1)$, of a sequence of dependent Brownian motions $\{\mathcal{W}_{tj}\}_{1\leq j\leq p}$, which may be of independent interest.

    \item[(c)] We further prove that either version of the max-$L_\infty$-type statistics is asymptotically independent of certain sum-$L_2$-type statistic under some mild conditions if $H_0$ holds, which forms the basis of the DMS method.
    Besides, we provide an asymptotic lower bound of the local power function of the proposed test due to the observation that the asymptotic independence also holds under a sequence of local alternative hypotheses. By doing so, we could immediately obtain the consistency of the DMS method and in addition present a detection regime for which the DMS outperforms both the max-$L_\infty$- and sum-$L_2$-procedures.
    The latter offers a new insight on adaptive testing methods beyond the consistency that investigated in the literature \citep{10.1111/rssb.12375,doi:10.1080/01621459.2021.1884562}.
\end{itemize}

The remainder of this paper is structured as follows. In Section \ref{sec:max}, we investigate asymptotic null distributions of two versions of the max-$L_{\infty}$-type statistics, namely, $M_{n,p}$ and $M^{\dagger}_{n,p}$. The proposed DMS procedure is presented in Section \ref{sec:DMS}, together with its asymptotic properties. Numerical studies are conducted in Section \ref{sec:num}. Section \ref{sec:final} concludes the paper, and all theoretical proofs are deferred to Supporting Information.

\textbf{Notations}: Throughout this paper, we use $\lesssim$, $\gtrsim$, ($\sim$) to denote (two-sided) inequalities invovling a multiplicative constant. For $a\in\mathbb{R}$, we denote by $\lfloor a\rfloor$ the lower integer part of $a$. For a set $\mathcal{A}$, we denote by $|\mathcal{A}|$ its cardinality, and by $\mathcal{A}^c$ its complementary. For a matrix $A$, let $\tr(A)$ be its trace. For a vector $a$, $\diag(a)$ represents a diagonal matrix whose diagonal elements take values in $a$.

\section{Max-$L_{\infty}$-aggregation}\label{sec:max}

We first investigate two versions of max-$L_\infty$-type statistics that were considered in the literature \citep{10.1214/15-aos1347,10.1111/rssb.12406,10.1111/rssb.12375}, to wit,
\begin{align}\label{M}
    M_{n,p}:=\max_{k=1,\ldots,n-1}\max_{j=1,\ldots,p}|C_{0, j}(k)|\ \text{and}\
    M^{\dagger}_{n,p}:=\max_{\lambda_n\leq k\leq n-\lambda_n}\max_{j=1,\ldots,p}|C_{0.5, j}(k)|,
\end{align}
respectively, where we recall that the CUSUM statistics $C_{\gamma, j}(k)$'s with $\gamma=0$ or 0.5 are defined in \eqref{CUSUM}, and $\lambda_n\in[1,n/2]$ is a pre-specified boundary removal parameter. Specification of the estimators $\widehat{\sigma}_j$'s in $C_{\gamma, j}(k)$'s will be discussed later.

If $p=1$, it is well-known that, as $n\to\infty$, $M_{n,1}$ converges in distribution to $\sup_{0\leq t\leq 1}B(t)$ under $H_0$, where $B$ denotes the standard Brownian bridge.
The power of $M_{n,1}$ may decay if the changepoint appears early or late.
To address this, many weighted versions have been introduced, among which the most prominent is $M^{\dagger}_{n,1}$ that is equivalent to the maximally selected likelihood ratio test assuming Gaussianity. The asymptotic null distribution of $M^{\dagger}_{n,1}$ depends on the boundary removal parameter $\lambda_n$. It can be verified that $M^{\dagger}_{n,1}$ converges in distribution to $\sup_{\lambda\leq t\leq 1-\lambda} B(t)/\sqrt{t(1-t)}$ if $\lambda_n/n\to \lambda\in(0,0.5)$.
However, if $\lambda_n=1$ or $\lambda_n/n\to 0$, $M^{\dagger}_{n,1}$ would diverge under $H_0$.
It is clear now that $M^{\dagger}_{n,1}$, with necessary normalization, converges in distribution to the extreme value distribution of Gumbel type due to \cite{MR74712}.
More discussions under $p=1$ can be found in \cite{csorgo1997limit} and \cite{10.1111/j.1467-9892.2012.00819.x}.

When $(n, p)\to\infty$ in the sense that both $n$ and $p$ diverge jointly, \cite{10.1214/15-aos1347} showed that $M_{n,p}$, appropriately normalized, weakly converges to the Gumbel distribution under certain decay conditions on the correlations if $H_0$ holds.
For $M^{\dagger}_{n,p}$, \cite{10.1111/rssb.12406} presented a valid approximation to its null distribution by using the multiplier bootstrap if
$\lambda_n/\log^7(np)\to\infty$ for sub-exponential data sequences.
However, little is known regarding the asymptotic null distribution of $M^{\dagger}_{n,p}$ as $(n, p)\to\infty$.
In this paper, we shall fill this gap and show that the Darling-Erd\"{o}s-type convergence also holds for $M^{\dagger}_{n,p}$, with suitable chosen $\lambda_n$, under $H_0$.
In addition, we are able to derive asymptotic null distributions of both $M^{\dagger}_{n,p}$ and $M_{n,p}$ under a more relaxed assumption on componentwise correlations among $p$ variables than that in \cite{10.1214/15-aos1347}.

Before proceeding any further, we introduce some notations and assumptions on the dependence of the noises $\epsilon_{ij}$'s both in time and across all variables. Suppose there exist measurable functions $g_j$'s such that $\epsilon_{ij} = g_j(e_{i},e_{i-1},\ldots)$, where $\{e_i\}_{i\in\mathbb{Z}}$ is a sequence of independent and identically distributed (i.i.d.) random variables. To measure temporal dependence, define for $q\geq 1$,
\begin{align*}
    a_i(q)=\limsup_{p\to\infty}\max_{j=1,\ldots,p}\|g_j(e_{i},e_{i-1},\ldots,e_{0},e_{-1},\ldots) - g_j(e_{i},e_{i-1},\ldots,e'_{0},e_{-1},\ldots)\|_{q},
\end{align*}
where $\{e'_{i}\}_{i\in\mathbb{Z}}$ is an independent copy of $\{e_i\}_{i\in\mathbb{Z}}$.
Let $\sigma_{jj'}=\lim_{n\to\infty}n^{-1}\E\left(\sum_{i=1}^n\sum_{i'=1}^n\epsilon_{ij}\epsilon_{i'j'}\right)$ be the long-run covariances, and denote $\sigma_{j}=\sigma_{jj}^{1/2}$. The componentwise correlations among $p$ variables can thus be defined as $\rho_{jj'}=\sigma_{jj'}/(\sigma_{j}\sigma_{j'})$. Denote $R=(\rho_{jj'})_{p\times p}$.
Let for some sequences $\delta_p>0$ and $\kappa_p>0$, $B_{p,j}=\{1\leq j'\leq p: |\rho_{jj'}|\geq\delta_p\}$ and $C_p=\{1\leq j\leq p: |B_{p,j}|\geq p^{\kappa_p}\}$.
To estimate the long-run variances $\sigma_{j}$'s, we consider Bartlett's estimators, to wit, $\widehat{\sigma}_{j}^{2}=\sum_{|\ell|\leq b_{n}}(n-\ell)^{-1}\sum_{i=\ell+1}^{n}\left(X_{\ell j}-\bar{X}_{j}\right)\left(X_{i-\ell, j}-\bar{X}_{j}\right)$ for some $b_{n}\geq 1$, where $\bar{X}_{j}=n^{-1}\sum_{i=1}^{n} X_{ij}$.

\begin{assumption}[Temporal dependence]\label{asmp:corr_temporal}
There exist some constants $q>4$ and $\mathfrak{a}>5/2$ such that $a_{i}(q)\lesssim i^{-\mathfrak{a}}$. In addition, $\liminf_{p\to\infty}\min_{j=1,\ldots,p}\sigma_{j}\geq\underline{\sigma}$ for some constant $\underline{\sigma}>0$.
\end{assumption}

\begin{assumption}[Componentwise correlations]\label{asmp:corr_spatial}
(i) $|\rho_{jj'}|\leq\varrho$ for $1\leq j\neq j'\leq p$ and some constant $\varrho\in (0,1)$; (ii) $|C_p|/p\to 0$ for some $\delta_p=o\{(\log p)^{-1}\}$ and $\kappa_p\to 0$, as $p\to\infty$.
\end{assumption}

\begin{remark}
Assumption \ref{asmp:corr_temporal} imposes a polynomial decay restriction on the temporal dependence, and was considered by \cite{10.1214/15-aos1347}. A huge variety of popular linear and nonlinear time series models meet Assumption \ref{asmp:corr_temporal} \citep{MR2172215,10.1214/15-aos1347}. In addition, it implies that $\sigma_{jj'}$'s are well defined \citep{10.1214/15-aos1347}.
Assumption \ref{asmp:corr_spatial}--(ii) demands the number of variables that are strongly-correlated (i.e., $|\rho_{jj'}|\geq\delta_p$) with many ($\geq p^{\kappa_p}$) other variables should not be too much ($o(p)$). It is more relaxed than the logarithmic decay assumption in \cite{10.1214/15-aos1347}, to wit, $\rho_{jj'}\lesssim\log^{-2-\zeta}(|j-j'|+2)$ for some $\zeta>0$. To see this, the latter requires a natural ordering among all variables that is not always easy to meet in practice.
\end{remark}

\begin{theorem}\label{null:Max}
Suppose $H_0$ and Assumptions \ref{asmp:corr_temporal}--\ref{asmp:corr_spatial} hold. {Assume $b_{n}\sim n^{\mathfrak{b}}$ for some $0<\mathfrak{b}<1$ and $p\lesssim n^{\nu}$ for some $0<\nu<\min\{p/2-2,(1-\mathfrak{b})p/2-1\}$.}
\begin{itemize}
\item[(i)] As $(n,p)\to\infty$,
\[
    \pr\left(M_{n,p}\leq u_p\{\exp(-x)\}\right) \to \exp\{-\exp(-x)\},
\]
where $u_p\{\exp(-x)\}=\sqrt{\{x+\log(2p)\}/2}$.

\item[(ii)] If $\lambda_n\sim n^{\lambda}$ for some $\lambda\in(0,1)$, then, as $(n,p)\to\infty$,
\[
    \pr\left(M^{\dagger}_{n,p}\leq \frac{x+D(p\log h_n)}{A(p\log h_n)}\right) \to \exp\{-\exp(-x)\},
\]
where $A(x)=\sqrt{2\log x}$, $D(x)=2\log x+2^{-1}\log\log x-2^{-1}\log\pi$ and $h_n=\left\{(\lambda_n/n)^{-1}-1\right\}^2$.
\end{itemize}
\end{theorem}

\begin{remark}
One of the key ingredients in our proof for Theorem \ref{null:Max} is the study of the asymptotic distribution of $\max_{j=1,\ldots,p} Z_{\gamma,j}$ with $\gamma=0$ or 0.5 if $X_i\sim N(0,\Sigma)$ independently with $\Sigma=(\sigma_{jj'})_{p\times p}$, where $Z_{0,j}=\max_{k=1,\ldots,n-1}|C_{0, j}(k)|$ and $Z_{0.5,j}=\max_{\lambda_n\leq k\leq n-\lambda_n}|C_{0.5, j}(k)|$. Then, under non-Gaussian scenario, we can apply the Gaussian approximation in conjunction with the truncation arguments \citep{MR3161448,10.1214/15-aos1347}. However, under Assumption \ref{asmp:corr_spatial}, the derivations are highly non-trivial even for independent Gaussian data sequences. \cite{10.1214/15-aos1347} studied the limit distribution of $\max_{j=1,\ldots,p} Z_{0,j}$ under the logarithmic decay assumption by supplying the blocking arguments \citep{MR362465}, which is infeasible under Assumption \ref{asmp:corr_spatial} due to the absence of the ordering of $p$ variables. To this end, we develop new tools and techniques based on the \textit{inclusion-exclusion principle} to make it reachable.
To be specific, let $u_p:=u_p\{\exp(-x)\}$. According to Step 1 in the proof of Proposition \ref{null:Mnp-Gaussian} in Supporting Information, it suffices to show that
\begin{align}\label{null:Max:proof}
    \pr\big(\max_{j\not\in C_p}Z_{0,j}>u_p\big)\to 1-\exp\{-\exp(-x)\},
\end{align}
By the inclusion-exclusion principle, we can construct sharp lower and upper bounds of $\pr\big(\max_{j\not\in C_p}Z_{0,j}>u_p\big)$. To wit, for any $k\geq 1$,
\[
    \sum_{t=1}^{2k}(-1)^{t-1}\alpha_t\leq \pr\big(\max_{j\not\in C_p}Z_{0,j}>u_p\big)\leq \sum_{t=1}^{2k+1}(-1)^{t-1}\alpha_t,
\]
where $\alpha_t=\sum \pr\big(Z_{0,j_1}>u_p,\ldots,Z_{0,j_t}>u_p\big)$ and the sum runs over all combinations $j_1,\ldots,j_t\in C_p^c$ such that $j_1<\cdots<j_t$. The main difficulties lie in the verification of the fact that $\alpha_t\to \frac{1}{t!}\exp(-tx)$, see Steps 3--6 in the proof of Proposition \ref{null:Mnp-Gaussian}. As a consequence, by letting $k\to\infty$ and using the Taylor expansion of the function $1-\exp(-x)$, \eqref{null:Max:proof} immediately follows.
Some intermediate conclusions on asymptotic properties of random variables $\sup_{0\leq t\leq 1}|\mathcal{W}_{tj}-t\mathcal{W}_{1j}|$ for $j=1,\ldots,p$ are demanded, see Lemmas \ref{lemma:B>=up}--\ref{lemma:Z0j_strong} in Supporting Information, where $\{\mathcal{W}_{tj}\}_{1\leq j\leq p}$ is a sequence of dependent Brownian motions.
Besides, we apply such ideas to study the limit distribution of $\max_{j=1,\ldots,p} Z_{0.5,j}$ in conjunction with an investigation of asymptotic properties on dependent random variables $\sup_{\Lambda_p\leq t\leq 1-\Lambda_p}\{t(1-t)\}^{-1/2}|\mathcal{W}_{tj}-t\mathcal{W}_{tj}|$, with suitably chosen $\Lambda_p=o(1)$.
\end{remark}

Based on Theorem \ref{null:Max}, we can easily obtain the p-value associated with either $M_{n,p}$ or $M^{\dagger}_{n,p}$, namely,
\begin{align*}
    {\rm p}_{M_{n,p}} &:= 1 - G\big(2M_{n,p}^2-\log(2p)\big)\ \text{or}\\
    {\rm p}_{M^{\dagger}_{n,p}} &:= 1 - G\big(A(p\log h_n)M^{\dagger}_{n,p}-D(p\log h_n)\big),
\end{align*}
respectively, where $G$ denotes the standard Gumbel distribution, i.e., $G(x)=\exp\{-\exp(-x)\}$.
If the p-value is below some pre-specified significant level, say $\alpha\in(0,1)$, then we rejected the null hypothesis that there is no changepoint in the data sequence.
It can be expected that either max-$L_\infty$-based testing procedure would be effective in detecting sparse and strong change signals.

\section{The DMS method}\label{sec:DMS}

\subsection{Sum-$L_2$-aggregation}
To detect dense but possibly weak changes, we consider a sum-$L_2$-based aggregation approach that was proposed by \cite{10.1214/17-aos1610} for multinomial data and \cite{10.1142/s201032631950014x} for sub-Gaussian data, respectively. To be specific, we use
\begin{align}\label{S}
    S_{n,p} := \sum_{1\leq k< n}\sum_{j=1,\ldots,p} C^2_{0.5, j}(k).
\end{align}
Recall that $R=(\rho_{jj'})_{p\times p}$.
Lemma \ref{null:Sum} restates the asymptotic null distribution of $S_{n,p}$ derived in \cite{10.1142/s201032631950014x}.
\begin{lemma}\label{null:Sum}
Assume that (i) $\epsilon_{ij}=\Sigma^{1/2}\varepsilon_{ij}$, where $\varepsilon_{ij}$ are i.i.d. sub-Gaussian variables; (ii) $\tr(R^4)=o\{\tr^2(R^2)\}$ as $p\to\infty$; (iii) $p/n^{3-\upsilon}\to 0$ for some $\upsilon>0$. Then, as $(n,p)\to\infty$,
\[
    \{S_{n,p}-(n+2)p\}/\sqrt{\var(S_{n,p})}\to N(0,1)
\]
in distribution, and $\var(S_{n,p}) = \left[\frac{2\pi^2-18}{3}n^2\tr(R^2)+\frac{15-\pi^2}{3}n\left\{\E(\varepsilon^\T R\varepsilon)^2-p^2\right\}\right]\{1+o(1)\}$.
Further, if $\widehat{\tr(R^2)}/\tr(R^2)\to 1$ and $\widehat{\E(\varepsilon^\T R\varepsilon)^2}/\E(\varepsilon^\T R\varepsilon)^2\to 1$ in probability, then $\{S_{n,p}-(n+2)p\}/V^{1/2}_{p,n}\to N(0,1)$ in distribution, where $V_{n,p}=\frac{2\pi^2-18}{3}n^2\widehat{\tr(R^2)}+\frac{15-\pi^2}{3}n\left\{\widehat{\E(\varepsilon^\T R\varepsilon)^2}-p^2\right\}$.
\end{lemma}

\begin{remark}
\cite{10.1142/s201032631950014x} suggested using $\widehat{\sigma}_j^2=\{2(n-1)\}^{-1}\sum_{i=2}^{n}(X_{ij}-X_{i-1,j})^2$ for $j=1,\ldots,p$. Further, $\tr(R^2)$ and $\E(\varepsilon^\T R\varepsilon)^2$ can be estimated by
\begin{align*}
    \widehat{\tr(R^2)} &=
    \frac{1}{4(n-3)} \sum_{i=1}^{n-3}\left\{
    (X_{i}-X_{i+1})^\T \widehat{D}_{(i,i+1,i+2,i+3)}^{-1} (X_{i+2}-X_{i+3})
    \right\}^2\ \text{and}\ \\
    \widehat{\E(\varepsilon^\T R\varepsilon)^2} &=
    \frac{1}{(n-2)} \sum_{i=1}^{n-2}\left\{
    (X_{i}-X_{i+1})^\T \widehat{D}_{(i,i+1,i+2)}^{-1} (X_{i+1}-X_{i+2})
    \right\}^2
    - 3 \widehat{\tr(R^2)},
\end{align*}
respectively, where for any $(i_1,\ldots,i_m)\subset\{1,\ldots,n\}$ with $m\geq 1$,
\[
\widehat{D}_{(i_1,\ldots,i_m)}=\diag\left(\widehat{\sigma}_{1(i_1,\ldots,i_m)}^{2}, \ldots, \widehat{\sigma}_{p(i_1,\ldots,i_m)}^{2}\right),
\]
and $\widehat{\sigma}_{j(i_1,\ldots,i_m)}^{2}=\{2|\mathcal{A}_m|\}^{-1}\sum_{i\in\mathcal{A}_m}(X_{ij}-X_{i-1,j})^2$ with $\mathcal{A}_m=\{2,\ldots,n\}\backslash\{i_1,\ldots,i_m\}$ for $j=1,\ldots,p$. With the usages of such difference-based estimators, Lemma \ref{null:Sum} holds.
\end{remark}

By Lemma \ref{null:Sum}, the p-value associated with $S_{n,p}$ is
\begin{align*}
    {\rm p}_{S_{n,p}} &:= 1 - \Phi\left(\{S_{n,p}-(n+2)p\}/V_{n,p}^{1/2}\right),
\end{align*}
where $\Phi$ is the cumulative distribution function (CDF) of $N(0,1)$. Again, small values of ${\rm p}_{S_{n,p}}$ indicate rejecting the null hypothesis.

\subsection{Adaptive strategy}

In practice, we seldom know whether the potential change signal is sparse or dense. To adapt to various alternative behaviors, we propose combining the max-$L_\infty$- and sum-$L_2$-based testing procedures. The key message in this paper is that the max-$L_\infty$- and sum-$L_2$-type statistics are asymptotically independent under some mild conditions if $H_0$ holds.
Following most works on adaptive changepoint testing such as \cite{10.1111/rssb.12375} and \cite{doi:10.1080/01621459.2021.1884562}, we restrict ourselves to the setting of independent observations. In such scenario, we propose to use the difference-based variance estimators, i.e., $\widehat{\sigma}_j^2=\{2(n-1)\}^{-1}\sum_{i=2}^{n}(X_{ij}-X_{i-1,j})^2$ for $j=1,\ldots,p$.
To ease the presentation and highlight key ideas in the proof, we also introduce the following assumptions.

\begin{assumption}\label{asmp:noise}
The noises $\epsilon_{ij}=\Sigma^{1/2}\varepsilon_{ij}$, where $\varepsilon_{ij}$ are i.i.d. sub-Gaussian variables, i.e., there exists a constant $\zeta>0$ such that $\E\{\exp(t\varepsilon_{ij})\}\leq \exp(\zeta t^2)$ for all $t\in\mathbb{R}$.
\end{assumption}

\begin{assumption}\label{asmp:corr_spatial-ev}
(i) $\liminf_{p\to\infty}\min_{j=1,\ldots,p}\sigma_{j}\geq\underline{\sigma}$ for some constant $\underline{\sigma}>0$;
(ii) $|\rho_{jj'}|\leq\varrho$ for $1\leq j\neq j'\leq p$ and some constant $\varrho\in (0,1)$;
(iii) There exist some constants $0<\underline{c}<\overline{c}<\infty$ such that $\underline{c}\leq\lambda_{\min}(R)\leq\lambda_{\max}(R)\leq\overline{c}$, where $\lambda_{\min}(R)$ and $\lambda_{\max}(R)$ denote the minimal and maximal eigenvalues of $R$.
\end{assumption}

\begin{remark}
Under Assumptions \ref{asmp:noise}--\ref{asmp:corr_spatial-ev}, Lemma \ref{null:Sum} trivially holds if $p/n^{3-\upsilon}\to 0$ for some $\upsilon>0$. In addition, Assumption \ref{asmp:corr_spatial-ev}--(iii) implies that Assumption \ref{asmp:corr_spatial}--(ii) holds, and thus it can be shown that Theorem \ref{null:Max} holds if $p\lesssim n^{\nu}$ for some $\nu>0$. To see this, let $R=Q^\T \Lambda Q$, where $Q=(q_{jj'})_{p\times p}$ is an orthogonal matrix and $\Lambda=\diag(\lambda_1,\ldots,\lambda_p)$ with $\lambda_j$'s being the eigenvalues of $R$. Observe that $\sum_{j'=1}^p \rho_{jj'}^2=\sum_{j'=1}^p q_{j'j}^2\lambda_{j'}^2\leq\overline{c}^2$. As a consequence, $|B_{p,j}|\delta_p^2\leq \sum_{j'=1}^p \rho_{jj'}^2\leq \overline{c}^2$. Taking $\delta_p=(\log p)^{-1-c}$ for some constant $c>0$ and select $\kappa_p=4(1+c)\log\log p/\log p\to 0$, we conclude that $|B_{p,j}|<p^{\kappa_p}$ for sufficiently large $p$. Hence $|C_p|=0$ and Assumption \ref{asmp:corr_spatial}--(ii) holds.
\end{remark}

\begin{theorem}\label{null:DMS}
Suppose $H_0$ and Assumptions \ref{asmp:noise}--\ref{asmp:corr_spatial-ev} hold. Assume $p\lesssim n^{\nu}$ for some $0<\nu<\min\{p/2-2,3-\upsilon\}$ with some $\upsilon>0$.
\begin{itemize}
\item[(i)] As $(n,p)\to\infty$, $M_{n,p}$ is asymptotically independent of $S_{n,p}$ in the sense that
\[
    \pr\Big(M_{n,p}\leq u_p\{\exp(-x)\}, \frac{S_{n,p}-(n+2)p}{V_{n,p}^{1/2}}\leq y\Big) \to \exp\{-\exp(-x)\}\cdot\Phi(y);
\]

\item[(ii)] If $\lambda_n\sim n^{\lambda}$ for some $\lambda\in(0,1)$, then, as $(n,p)\to\infty$, $M^{\dagger}_{n,p}$ is asymptotically independent of $S_{n,p}$ in the sense that
\[
    \pr\Big(M^{\dagger}_{n,p}\leq \frac{x+D(p\log h_n)}{A(p\log h_n)}, \frac{S_{n,p}-(n+2)p}{V_{n,p}^{1/2}}\leq y\Big) \to \exp\{-\exp(-x)\}\cdot\Phi(y).
\]
\end{itemize}
\end{theorem}

\begin{remark}
Proof of Theorem \ref{null:DMS} is first carried under Gaussian data sequences, i.e., $X_i\sim N(0,\Sigma)$, by leveraging the \textit{inclusion-exclusion principle} and projection arguments. Then it is followed by the Gaussian approximation in conjunction with a smooth approximation of the maximum function \citep{MR3161448} to handle sub-Gaussian data.
To fix ideas, we consider Theorem \ref{null:DMS}--(i) under Gaussianity, and the goal is to demonstrate asymptotic independence between $\max_{j=1,\ldots,p} Z_{0,j}$ with $Z_{0,j}=\max_{k=1,\ldots,n-1}|C_{0, j}(k)|$ and $\widetilde{S}_{n,p}:=\{S_{n,p}-(n+2)p\}/{V_{n,p}^{1/2}}$. For any fixed $x,y\in\mathbb{R}$, define $A_p(x)=\{\widetilde{S}_{n,p}\leq x\}$ and $B_j(y)=\{Z_{0,j}>u_p\{\exp(-y)\}\}$ for $j=1,\ldots,p$. By Theorem \ref{null:Max} and Lemma \ref{null:Sum}, it suffices to show that
\[
    \pr\Big(\bigcup_{j=1}^p A_p(x)B_j(y)\Big)\to\Phi(x)\cdot\{1-\exp\{-\exp(-y)\}.
\]
By applying the inclusion-exclusion principle to $\pr\left(\bigcup_{j=1}^p A_p(x)B_j(y)\right)$ and $\pr\left(\bigcup_{j=1}^p B_j(y)\right)$, we can show
\begin{align}\label{null:DMS:proof}
    \pr\Big(\bigcup_{j=1}^p A_p(x) B_j(y)\Big)
    \leq \pr\left(A_p(x)\right)\pr\Big(\bigcup_{j=1}^p B_j(y)\Big) + \sum_{t=1}^{2k} U(p,t) + V(p,2k+1),
\end{align}
where for each $t\geq 1$,
\begin{align*}
    U(p,t) &:= \sum_{1\leq j_1<\cdots<j_t\leq p}\left|\pr\left(A_p(x) B_{j_1}(y)\cdots B_{j_t}(y)\right) - \pr\left(A_p(x)\right)\cdot\pr\left(B_{j_1}(y)\cdots B_{j_t}(y)\right)\right|
\end{align*}
and $V(p,t):=\sum_{1\leq j_1<\cdots<j_t\leq p}\pr\left(B_{j_1}(y)\cdots B_{j_t}(y)\right)$. It can be verified that for each $t\geq 1$, $U(p,t)\to 0$ by using the projection method (cf. Lemma \ref{lem:indep:A-B} in Supporting Information) and $V(p,t)\to \frac{1}{t!}\exp(-tx/2)$ (see the proof of Theorem \ref{null:DMS}), as $(n,p)\to \infty$. By letting $k\to\infty$, it can be concluded from \eqref{null:DMS:proof} that $\limsup_{(n,p)\to\infty}\pr\Big(\bigcup_{j=1}^p A_p(x)B_j(y)\Big)\leq\Phi(x)\cdot\{1-\exp\{-\exp(-y)\}\}$. Similarly, we have $\liminf_{(n,p)\to\infty}\pr\Big(\bigcup_{j=1}^p A_p(x)B_j(y)\Big)\geq\Phi(x)\cdot\{1-\exp\{-\exp(-y)\}\}$. Hence the conclusion follows.
\end{remark}

According to Theorem \ref{null:DMS}, we suggest combining the corresponding p-values by using Fisher's method \citep{MR312634,MR375577}. To wit,
\begin{align*}
    {\rm p}_{M,S} &:= 1 - F_{\chi^2_4}\Big(-2(\log {\rm p}_{M_{n,p}} + \log {\rm p}_{S_{n,p}})\Big)\ \text{or}\\
    {\rm p}_{M^{\dagger},S} &:= 1 - F_{\chi^2_4}\Big(-2(\log {\rm p}_{M^{\dagger}_{n,p}} + \log {\rm p}_{S_{n,p}})\Big),
\end{align*}
where $F_{\chi^2_4}$ is the CDF of the chi-squared distribution with 4 degrees of freedom.
The rationality is that either $-2(\log {\rm p}_{M_{n,p}} + \log {\rm p}_{S_{n,p}})$ or $-2(\log {\rm p}_{M^{\dagger}_{n,p}} + \log {\rm p}_{S_{n,p}})$ converges in distribution to $F_{\chi^2_4}$ under $H_0$ due to Theorem \ref{null:DMS}.
Then either ${\rm p}_{M,S}$ or ${\rm p}_{M^{\dagger},S}$ can be used as the final p-value for testing $H_0$. If it is less than some pre-specified significant level $\alpha\in(0,1)$, then we reject $H_0$.

\begin{remark}
The sum-$L_2$-type statistics were proposed by \cite{10.1214/17-aos1610} and \cite{10.1142/s201032631950014x}, and have been shown to admit an asymptotic Gaussian distribution under $H_0$ (cf. Lemma \ref{null:Sum}). The authors also considered a power-enhancement term \citep{10.3982/ecta12749} based on the max-$L_\infty$-type statistic to form a new testing procedure. However, the rationality is rather different from ours. To wit, under $H_0$, the power-enhancement term does not alter the asymptotic Gaussianity of the test statistic. Under $H_1$, it would generally diverge and dominate the testing procedure if the change signal is sparse and strong, and thus enhance the testing power.
In contrast, we exploit the joint distribution of both max-$L_\infty$- and sum-$L_2$-type statistics to adapt to different levels of sparsity. Our numerical studies reveal that the DMS method outperforms such power-enhancement ones in the sparse regime since it leverages the information from the max-$L_\infty$-type statistic more efficiently.
\end{remark}

\begin{remark}
Sparsity-adaptive changepoint testing has been widely discussed in the literature, as reviewed in Section \ref{sec:introduction}.
To make comparisons among state of the art approaches, we consider the max-$L_{(s,q)}$-type statistic \citep{10.1111/rssb.12375} which aggregates the CUSUMs or their variants, say $Z_{j}(k)$, over $p$ dimensions by an adjusted $L_q$ norm, namely,
$T_{q,s} := \max_{k}\big\{\sum_{j=1}^{s}|Z_{(j)}(k)|^q\big\}^{1/q}$,
where $|Z_{(1)}(k)|\geq\cdots\geq|Z_{(p)}(k)|$ is the order statistics of $\{|Z_{j}(k)|\}_{j=1}^p$.
\cite{10.1214/18-aos1740} and \cite{10.1214/20-aos1994} proposed scan statistics that searches for the maximum of appropriately normalized $T_{2,s}$'s over a range of $s$ values with well-calibrated normalizing factors. In addition, they showed that their methods can achieve the minimax testing rate for Gaussian data sequences; however, the corresponding asymptotic distributions are unknown.
\cite{10.1111/rssb.12375} suggested combining different $T_{q,s}$ over $q$'s with a fixed $s=s_0$ by taking the minimum of associated p-values, which is obtained via the multiplier bootstrap approximation.
The validity of bootstrap approximation was also justified, without the need of knowing the complicated joint distribution of all $T_{q,s_0}$'s over $q$'s.
For practical applications, they recommended using $q\in\{1,2,3,4,5,\infty\}$ and $s_0=p/2$ (although this choice results in a diverging $s$).
\cite{doi:10.1080/01621459.2021.1884562} investigated the asymptotic distribution of a specialized max-$L_q$-type statistic (viz., a self-normalized U-statistic) with an even $q$,
and further justified the asymptotic independence of a finite number of the statistics indexed by different even $q$'s,
under the null hypothesis.
They then came up with the minimum p-value combination to aggregate different $q$ level statistics.
Notice that they excluded the statistic with $q=\infty$. In practice, they recommended using $q\in\{2,6\}$.
\end{remark}

The proposed DMS method only involves the max-$L_\infty$- and sum-$L_2$-type statistics, which can be computed in $O(np)$ time.
Due to the asymptotic independence of the two statistics (cf. Theorem \ref{null:DMS}) under $H_0$ and the associated reachable p-values (cf. Theorem \ref{null:Max} and Lemma \ref{null:Sum}),
it is very convenient to combine both testing procedures.
Consequently, the DMS is simple to use and computationally efficient;
it is free of well-calibrated thresholds or any numerical approximations such as by using multiplier bootstrap or by simulating Gaussian analogues.
As we will show in Section \ref{sec:num}, the power of the DMS is comparable or sometimes outperforms existing approaches; thus the DMS is also effective.

To theoretically investigate the power property of the DMS method, we consider the following sequence of local alternative hypotheses
\begin{align}\label{H1}
    H_{1;n,p}: |\mathcal{A}|=o\{p/(\log \log p)^2\}\ \text{and}\ \Delta_{S}=O\{\var^{1/2}(S_{n,p})\},
\end{align}
where $\mathcal{A}=\{1\leq j\leq p: \delta_j\neq 0\}$ is the support of $\delta$ and
\[
    \Delta_{S}=\sum_{\tau=1}^{n-1}\sum_{j=1}^{p}\frac{\tau(n-\tau)}{n}\Big\{
    \frac{1}{\tau}\sum_{i=1}^{\tau}\E(X_{ij})-\frac{1}{n-\tau}\sum_{i=\tau+1}^{n}\E(X_{ij})/\sigma_{j}^{2}
    \Big\}^2
\]
represents the scaled change signal corresponding to the sum-$L_2$-type statistic $S_{n,p}$.

\begin{theorem}\label{alter:DMS}
Suppose $H_{1;n,p}$ and Assumptions \ref{asmp:noise}--\ref{asmp:corr_spatial-ev} hold. Assume $p\lesssim n^{\nu}$ for some $0<\nu<\min\{p/2-2,3-\upsilon\}$ with some $\upsilon>0$.
\begin{itemize}
\item[(i)] As $(n,p)\to\infty$, $M_{n,p}$ is asymptotically independent of $S_{n,p}$ in the sense that
\[
    \pr\Big(M_{n,p}\leq u_p\{\exp(-x)\}, \frac{S_{n,p}-E(S_{n,p})}{V_{n,p}^{1/2}}\leq y\Big) - \pr\left(M_{n,p}\leq u_p\{\exp(-x)\}\right) \pr\Big(\frac{S_{n,p}-E(S_{n,p})}{V_{n,p}^{1/2}}\leq y\Big) \to 0;
\]
\item[(ii)] If $\lambda_n\sim n^{\lambda}$ for some $\lambda\in(0,1)$, then, as $(n,p)\to\infty$, $M^{\dagger}_{n,p}$ is asymptotically independent of $S_{n,p}$ in the sense that
\[
    \pr\Big(M^{\dagger}_{n,p}\leq v_{p,n}(x), \frac{S_{n,p}-E(S_{n,p})}{V_{n,p}^{1/2}}\leq y\Big) - \pr\left(M^{\dagger}_{n,p}\leq v_{p,n}(x)\right) \pr\Big(\frac{S_{n,p}-E(S_{n,p})}{V_{n,p}^{1/2}}\leq y\Big) \to 0,
\]
where $v_{p,n}(x):=\{x+D(p\log h_n)\}/{A(p\log h_n)}$.
\end{itemize}
\end{theorem}

\begin{remark}
Theorem \ref{alter:DMS} extends Theorem \ref{null:DMS} under the local alternatives $H_{1;n,p}$, which generally requires that the number of variables that occur changes should be not too much, i.e., $|\mathcal{A}|=o\{p/(\log \log p)^2\}$. In conjunction with $\Delta_{S}=O\{\var^{1/2}(S_{n,p})\}$, we can restrict ourselves to $S_{n,p}$ on $\mathcal{A}^c$, say $S_{n,p}(\mathcal{A}^c) := \sum_{1\leq k< n}\sum_{j\not\in\mathcal{A}} C^2_{0.5, j}(k)$, which is provably asymptotic independent of either $M_{n,p}$ or $M^{\dagger}_{n,p}$ on $\mathcal{A}$ by following similar arguments used in the proof of Theorem \ref{null:DMS}.
\end{remark}

How does the proposed DMS test perform compared to a single use of the max-$L_\infty$- or sum-$L_2$-type test? To ease notations, let $M$ be either $M_{n,p}$ or $M^{\dagger}_{n,p}$, and $S=S_{n,p}$. Denote the associated p-values by ${\rm p}_M$ and ${\rm p}_S$, respectively.
For a pre-specified significant level $\alpha\in(0,1)$, let $\beta_{M,\alpha}$ and $\beta_{S,\alpha}$ be the corresponding power functions based on $M$ and $S$, respectively.
According to \cite{MR312634,MR375577}, the power of Fisher's p-value combination-based test, say $\beta_{\alpha}$, should be larger than that of the test based on $\min\{p_{M},p_{S}\}$ (referred to as the minimal p-value combination), say $\beta_{M\wedge S, \alpha}$. Due to Theorem \ref{alter:DMS} together with the inclusion-exclusion principle, we have
\begin{align}\label{power}
    \beta_{\alpha}\geq \beta_{M\wedge S, \alpha}\geq \beta_{M,\alpha/2}+\beta_{S,\alpha/2}-\beta_{M,\alpha/2}\beta_{S,\alpha/2}.
\end{align}
For small $\alpha$, the difference between $\beta_{M,\alpha}$ and $\beta_{M,\alpha/2}$ should be small, and the same fact applies to $\beta_{S,\alpha}$. Consequently, by \eqref{power}, the power of the adaptive DMS test could be larger than that of either max-$L_\infty$- or sum-$L_2$-type test by a large margin.

\section{Numerical studies}\label{sec:num}

In this section, we investigate the finite-sample performance of the proposed DMS method. A collection of changepoint testing approaches based on various ways of aggregating the CUSUM statistics $C_{\gamma,j}(k)$'s or their variants are considered as benchmarks. To wit, with $\gamma=0$, the max-$L_\infty$-type procedure based on $M_{n,p}$ \citep{10.1214/15-aos1347} is included and we adopt the limit null distribution according to Theorem \ref{null:Max}--(i) in this paper; this method is termed as Max(0).
We also cover the adaptive procedure proposed by \cite{10.1111/rssb.12375} that is based on $T_{q,s_0}=\max_{\lambda\leq k\leq n-\lambda}\big\{\sum_{j=1}^{s_0}|C_{0, (j)}(k)|^q\big\}^{1/q}$ over $q\in\{1,2,3,4,5,\infty\}$ with $s_0=p/2$, where $q$'s and $s_0$ are chosen according to the authors' recommendation. The distribution of the involved test statistic is approximated by the multiplier bootstrap, with a bootstrap replication number $B=500$. This method appears to perform very well over a range choices of competitors, see Section 4 in \cite{10.1111/rssb.12375}, and thus is considered here as one of many state of the art adaptive approaches. We term it as LZZL.
With $\gamma=0.5$, we first cover a max-$L_\infty$- and sum-$L_2$-type statistics, i.e., $M^{\dagger}_{n,p}$ investigated in this paper and $S_{n,p}$ proposed by \cite{10.1142/s201032631950014x}, respectively.
The limit null distributions are specified according to Theorem \ref{null:Max}--(ii) and Lemma \ref{null:Sum}, respectively.
The corresponding approaches are referred to as Max(0.5) and Sum, respectively.
A bootstrap-based method in conjunction with $M^{\dagger}_{n,p}$ proposed by \cite{10.1111/rssb.12406} is also covered and termed as YC; as recommended by the authors, we set the number of bootstrap replications as $B=200$.
We in addition consider a power-enhanced testing approach proposed by \cite{10.1142/s201032631950014x}, which advocated the usage of the test statistic
\begin{align*}
    \{S_{n,p}-(n+2)p\}/V^{1/2}_{p,n} + c_{n,p}V^{1/2}_{p,n}\ind{M^{\dagger}_{n,p}>h_{n,p}}.
\end{align*}
As per the authors' suggestion, we specify $h_{n,p}=\sqrt{\{2\log(np)\}^{1.1}}$ and $c_{n,p}=100$. This method is termed as WZWY.
For the proposed DMS method, we use DMS(0) and DMS(0.5) to distinguish the scenarios with $\gamma=0$ and $\gamma=0.5$, respectively.
Table \ref{tab:benchmark} presents an overview of the above testing approaches that will be evaluated, together with the specification of involved nuisance parameters.
In particular, the LZZL, Max(0.5), YC, WZWY and DMS(0.5) approaches demand a boundary removal parameter (i.e., $\lambda_n$ or $\lambda$), which are set to be the same value for a fair comparison, for example, $\lambda_n=\lambda=40$ for $n=200$. We will fix $n=200$ for illustration.

\begin{table}[!ht]
\caption{An overview of testing approaches considered in simulated examples.}
\centering
\begin{tabular}{ccccc}
\toprule
Method & Critical value & Nuisance parameters & Related literature\\
\midrule
\multicolumn{4}{l}{{$\gamma = 0$}}
\\
\medskip
Max(0) & Asymptotic distribution & None & Our paper; \cite{10.1214/15-aos1347}\\
\smallskip

LZZL & Bootstrap & \begin{tabular}{@{}c@{}}$q\in\{1,2,3,4,5,\infty\}$,\\ $s_0=p/2$, $\lambda=40$, $B=500$\end{tabular} & \cite{10.1111/rssb.12375}\\
\smallskip

DMS(0) & Asymptotic distribution & None & Our paper\\
\smallskip

\\
\multicolumn{4}{l}{{$\gamma = 0.5$}}\\
\medskip
Max(0.5) & Asymptotic distribution & $\lambda_n=40$ & Our paper\\
\smallskip

YC & Bootstrap & $\lambda_n=40$, $B=200$ & \cite{10.1111/rssb.12406}\\
\smallskip

Sum & Asymptotic distribution & None & \cite{10.1142/s201032631950014x}\\
\smallskip

WZWY & Asymptotic distribution & \begin{tabular}{@{}c@{}}$\lambda_n=40$, $c_{n,p}=100$,\\ $h_{n,p}=\sqrt{\{2\log(np)\}^{1.1}}$\end{tabular} & \cite{10.1142/s201032631950014x}\\
\smallskip

DMS(0.5) & Asymptotic distribution & $\lambda_n=40$ & Our paper\\
\bottomrule
\end{tabular}
\label{tab:benchmark}
\end{table}

To generate the data, we set $\mu_0=0$ and two scenarios for the noises $\epsilon_{ij}=\Sigma^{1/2}\varepsilon_{ij}$ with $\Sigma=(\sigma_{jj'})_{p\times p}$ are considered, i.e,
\begin{itemize}
    \item[(I)] $\varepsilon_{ij}\sim N(0,1)$ independently and $\sigma_{jj'}=0.5^{|j-j'|}$ for $1\leq j,j'\leq p$;
    \item[(II)] $\varepsilon_{ij}$ are i.i.d. from $t$-distribution with degree of freedom 5 and are standardized such that $\var(\varepsilon_{ij})=1$, and $\Sigma$ admits a blocked diagonal structure, i.e., $\sigma_{jj}=1$, $\sigma_{jj'}=0.5$ for $5(k-1)<j\neq j'\leq 5k$ ($k=1,\ldots,\lfloor d/5\rfloor$) and $\sigma_{jj'}=0$ otherwise.
\end{itemize}
To account for different levels of the sparsity under alternatives, we set $\delta_{j}=\sqrt{\Delta/k}$ for $j=1,\ldots,k$ and $\delta_{j}=0$ otherwise. The magnitudes of $\delta_j$'s are chosen such that $\|\delta\|^2=\Delta$; as a result, the power curve of the Sum method would be roughly flat against $k$ and thus it could be set as a benchmark to facilitate the comparison.
To reflect the impact of the location of changepoint under alternatives, we vary $\tau$ over $\tau/n\in\{0.5,0.25\}$.
In all examples, the empirical sizes or power are calculated based on 1,000 replications and the nominal significance level is set to be 5\%.

\begin{table}[!ht]
    \caption{Observed sizes (in \%) of various tests carried out at the 5\% nominal level under Scenarios (I)--(II).}
    \centering
    \begin{tabular}{lcccccccc}
        \toprule
        $p$ & Max(0) & LZZL & DMS(0) & Max(0.5) & YC & Sum & WZWY & DMS(0.5)\\
        \midrule
        \multicolumn{9}{l}{Scenario (I)}\\
        \hspace{1em}100 & 4.3 & 8.0 & 6.5 & 3.0 & 3.7 & 6.4 & 6.5 & 6.2\\
        \hspace{1em}200 & 3.5 & 6.3 & 5.7 & 3.3 & 3.7 & 5.6 & 5.6 & 5.7\\
        \hspace{1em}300 & 4.0 & 5.4 & 5.7 & 2.9 & 4.1 & 5.2 & 5.2 & 4.8\\
        \\
        \multicolumn{9}{l}{Scenario (II)}\\
        \hspace{1em}100 & 5.0 & 6.0 & 5.8 & 3.7 & 2.0 & 4.8 & 5.2 & 5.8\\
        \hspace{1em}200 & 4.7 & 5.6 & 6.5 & 5.5 & 1.5 & 4.4 & 4.8 & 6.3\\
        \hspace{1em}300 & 5.5 & 5.4 & 6.4 & 5.2 & 0.6 & 4.1 & 4.7 & 5.8\\
        \bottomrule
    \end{tabular}
    \label{tab:sizes}
\end{table}

To evaluate the size performance, we range $p$ over $\{100,200,300\}$. Table \ref{tab:sizes} depicts the empirical sizes of various tests considered in Table \ref{tab:benchmark} under Scenarios (I)--(II). It can be seen that the proposed Max(0), Max(0.5), DMS(0) and DMS(0.5) methods can roughly maintain the nominal significance level, which demonstrates the validity of Theorems \ref{null:Max} and \ref{null:DMS}. The Max(0.5) seems slightly conservative under Scenario (I); however, this is improved for DMS(0.5) that leverages both maximum- and summation-type aggregations.
The YC method also mitigates this issue; however, it becomes much conservative under Scenario (II).
The Sum and WZWY perform very well due to fine Gaussian approximations.
Besides, the LZZL method behaves well as expected and may encounter a slight size inflation for small $p$ under Scenario (I).
In conclusion, all methods have a satisfactory performance under the null hypothesis.

\begin{figure}[!ht]
    \centering
    \includegraphics[width=\linewidth]{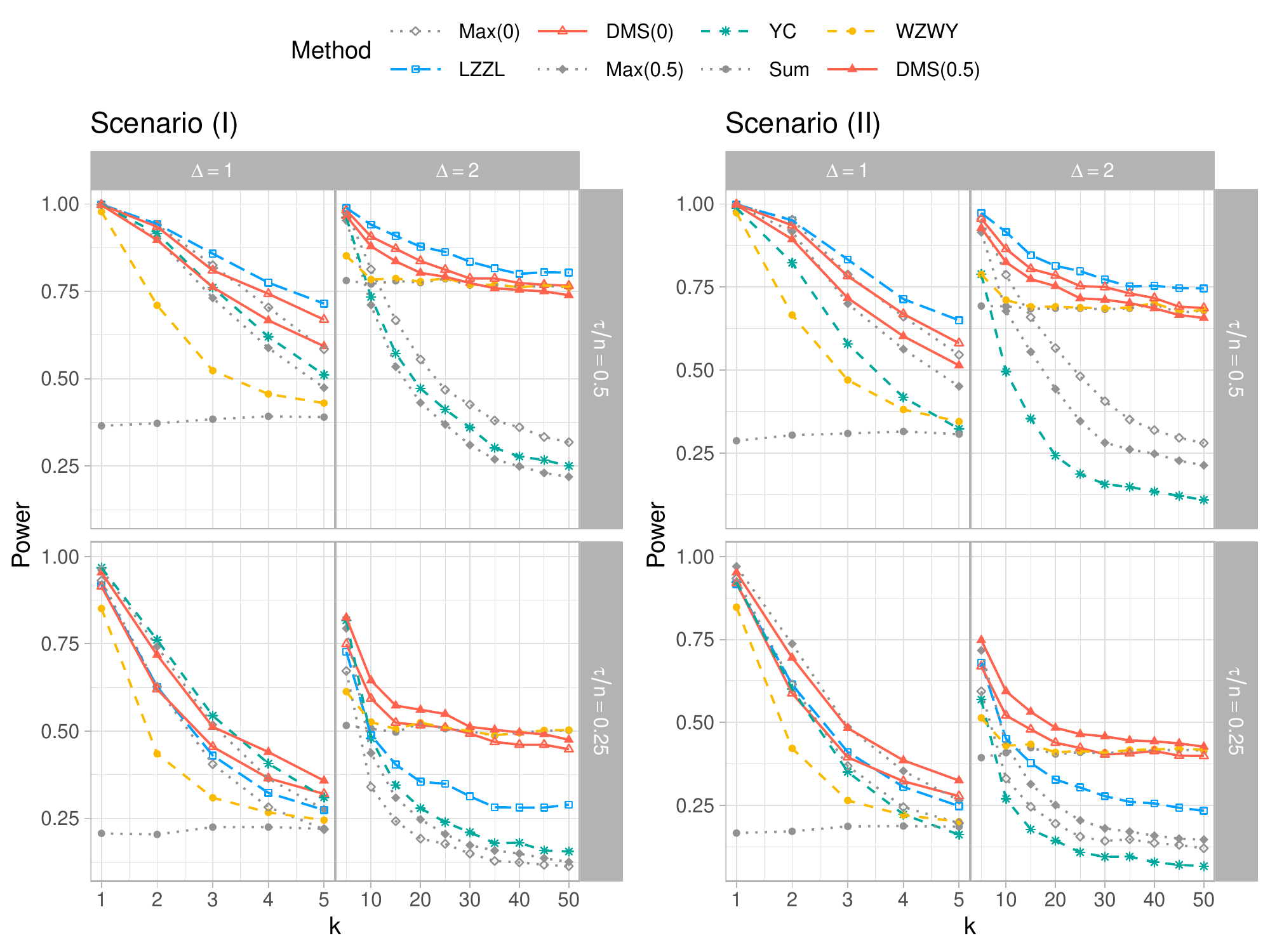}
    \caption{Observed power (in \%) of various tests carried out at the 5\% nominal level under Scenarios (I)--(II).}
    \label{fig:pow}
\end{figure}

Figure \ref{fig:pow} presents the empirical power of the tests we evaluated under Scenarios (I)--(II) with $p=200$. For a better visualization, we set the signal strength $\Delta=1$ if the sparsity level $k\in\{1,2,\ldots,5\}$ and set $\Delta=2$ if $k\in\{5,10,\ldots,50\}$, which can be roughly regarded as a sparse regime and a moderately sparse to dense regime. As a result, each power curve in Figure \ref{fig:pow} encounters a break at $k=5$. The adjustment for the signal strength is reasonable as the minimax testing rate may exhibits a phase transition against the sparsity level \citep{10.1214/18-aos1740,10.1214/20-aos1994}. Leaving it aside, we set the Sum method as a benchmark, since its power curve is flat in both regimes.
As expected, under every scenario, we first observe that the Max(0), Max(0.5) and YC methods outperform the Sum method for sparse change signals, while they fall behind the Sum in moderately sparse or dense regime.
In contrast to the Sum, the WZWY indeed enhances the power in the sparse regime and coincides with the Sum when raising $k$.
The adaptive procedure LZZL performs the best over all sparsity levels in both scenarios when $\tau/n=0.5$, i.e., the changepoint is in the middle; while the proposed adaptive DMS methods come second.
However, if the changepoint appears earlier, i.e., $\tau/n=0.25$, then the proposed DMS methods surpass the LZZL, especially in the dense regime.
In addition, the approaches with $\gamma=0$ outmatch those with $\gamma=1$ when $\tau/n=0.5$, while those with $\gamma=1$ come first when $\tau/n=0.25$.

In conclusion, the proposed DMS methods are adaptive to different levels of sparsity, and are comparable or sometimes outperform the LZZL method. Besides, the DMS is computationally efficient compared to the LZZL, which requires multiple replications of multiplier bootstrap.

\section{Concluding remarks}\label{sec:final}

In this paper, we propose a simple, fast yet effective procedure for adaptive changepoint testing. The key message here is that the max-$L_\infty$-type and sum-$L_2$-type test statistics are asymptotically independent under some conditions. We can then perform Fisher's p-value combination for the two separate p-values to form a new test, which is adaptive to different sparsity of change signals as demonstrated both theoretically and by numerical studies.

Although our discussion focuses on the mean change aspect, we believe it can be generalized to other scenarios including but not restricted to detecting changes in covariance matrix or regression coefficients. For example, we can exploit the framework of U-statistic \citep{10.1111/rssb.12375} or the score-based framework \citep{MR4065168}, and transform the problem into a general mean domain.

We have derived the asymptotic distribution for two versions of max-$L_\infty$-type statistic for both temporally and componentwisely dependent high-dimensional time series. The asymptotic independence of max-$L_\infty$-type and sum-$L_2$-type test statistics are investigated under the assumption that the data are independently collected over time. This temporal independence assumption is frequently considered in the literature of changepoint testing, see, for example, \cite{10.1111/rssb.12375}, \cite{10.1111/rssb.12406} and \cite{doi:10.1080/01621459.2021.1884562}, under which the proof of asymptotic independence is rather highly non-trivial. Extension to a more complicated scenario that allows temporal dependence certainly deserves further researches. A related issue is does the asymptotic normality still holds for the sum-$L_2$-type statistic for temporally dependent high-dimensional time series (e.g., under Assumptions \ref{asmp:corr_temporal}--\ref{asmp:corr_spatial}). We leave them as future work.

\section*{Supplementary material}

The supplementary material collects all proofs of Theorems \ref{null:Max}, \ref{null:DMS} and \ref{alter:DMS}, together with some necessary lemmas.
An R package, DMS, that implements the proposed DMS method can be found in the attached file DMS\_0.0.0.9000.tar.gz at this moment. A demonstration example is also included therein.

%

\newpage
\setcounter{section}{0}
\setcounter{equation}{0}
\setcounter{lemma}{0}
\setcounter{proposition}{0}
\renewcommand{\theequation}{S.\arabic{equation}}
\renewcommand{\thelemma}{S.\arabic{lemma}}
\renewcommand{\theproposition}{S.\arabic{proposition}}
\def\thetable{S\arabic{table}}
\def\thefigure{S\arabic{figure}}

\begin{center}
{\Large\textbf{Supplementary material for ``Computationally efficient and data-adaptive change point inference in high dimensions''}}
\end{center}

The supplementary material collects all proofs of Theorems \ref{null:Max}, \ref{null:DMS} and \ref{alter:DMS}, together with some necessary lemmas.
An R package, DMS, that implements the proposed DMS method can be found in the attached file DMS\_0.0.0.9000.tar.gz at this moment. A demonstration example is also included therein.

\textbf{Notations}: Throughout this paper, we use $\lesssim$, $\gtrsim$, ($\sim$) to denote (two-sided) inequalities invovling a multiplicative constant. We write $X\stackrel{d}{=} Y$ if two random variables $X$ and $Y$ have the same distribution.
For $a\in\mathbb{R}$, we denote by $\lfloor a\rfloor$ the lower integer part of $a$. For a set $\mathcal{A}$, we denote by $|\mathcal{A}|$ its cardinality, and by $\mathcal{A}^c$ its complementary. For a matrix $A$, let $\tr(A)$ be its trace. For a vector $a$, $\diag(a)$ represents a diagonal matrix whose diagonal elements take values in $a$.

Recall that, with $\gamma=0$ or 0.5, the CUSUM statistics are defined as
\begin{align*}
    C_{\gamma, j}(k) = \left\{\frac{k}{n}\left(1-\frac{k}{n}\right)\right\}^{-\gamma}\frac{1}{\sqrt{n}}\left(S_{kj}-\frac{k}{n}S_{nj}\right)/\widehat{\sigma}_{j},
\end{align*}
for $k=1,\ldots,n-1$ and $j=1,\ldots,p$, where $S_{kj}=\sum_{i=1}^{k}X_{ij}$ and $\widehat{\sigma}_{j}$'s are some estimators of $\sigma_j$'s.

\section{Proof of Theorem \ref{null:Max}}

We first study asymptotic null distributions of the max-$L_\infty$-type statistics $M_{n,p}$ and $M^{\dagger}_{n,p}$, to wit,
\begin{align*}
    M_{n,p}=\max_{k=1,\ldots,n-1}\max_{j=1,\ldots,p}|C_{0, j}(k)|\ \text{and}\
    M^{\dagger}_{n,p}=\max_{\lambda_n\leq k\leq n-\lambda_n}\max_{j=1,\ldots,p}|C_{0.5, j}(k)|,
\end{align*}
respectively, for independent observations that follow a Gaussian distribution; without loss of generality, we assume $X_i\sim N(0,\Sigma)$ for $i=1,\ldots,n$, where $\Sigma=(\sigma_{jj'})_{p\times p}$. We in addition replace each $\widehat{\sigma}_j$ by $\sigma_j$ in $C_{\gamma, j}(k)$ and assume that $\sigma_j=1$, for $j=1,\ldots,p$. In such settings,
\begin{align*}
   M_{n,p} &\stackrel{d}{=}\max_{j=1,\ldots,p}\sup_{1\leq k\leq n-1}|\mathcal{W}_{k/n,j}-(k/n)\mathcal{W}_{1j}|\ \text{and}\\
   M^{\dagger}_{n,p} &\stackrel{d}{=}\max_{j=1,\ldots,p}\max_{\lambda_n\leq k\leq n-\lambda_n}\{(k/n)(1-k/n)\}^{-1/2}|\mathcal{W}_{k/n,j}-(k/n)\mathcal{W}_{1j}|,
\end{align*}
where $\{\mathcal{W}_{tj}\}_{1\leq j\leq p}$ denotes a sequence of dependent Brownian motions.
Later in Section \ref{M-non-Gaussian} we will show that if $\widehat{\sigma}_j$'s are consistent estimators of $\sigma_j$'s, each corresponding asymptotic null distribution does not change.
Additionally, in Section \ref{M-non-Gaussian}, we will apply the Gaussian approximation in conjunction with the truncation arguments to study asymptotic null distributions of $M_{n,p}$ and $M^{\dagger}_{n,p}$ for non-Gaussian data sequences.

\subsection{For Gaussian data sequences}

\subsubsection{Asymptotic null distribution of $M_{n,p}$}

For any fixed $x$, let $w=\exp(-x)$ and $u_p:=u_p(w)=\sqrt{\{x+\log(2p)\}/2}$.
Let $\mathcal{B}=\sup_{0\leq t\leq 1}|\mathcal{W}_t-t\mathcal{W}_1|$ and $\mathcal{B}_n=\sup_{0\leq k\leq n}|\mathcal{W}_{k/n}-(k/n)\mathcal{W}_1|$, where $\mathcal{W}_t$ is a standard Brownian motion.
Let $Z_{0,j}=\max_{1\leq k<n}|C_{0,j}(k)|$ and thus $Z_{0,j}\stackrel{d}{=}\mathcal{B}_n$, for each $j=1,\ldots,p$.

\begin{lemma}\label{lemma:B>=up}
It holds that
\begin{align}\label{B>=up}
    \pr\left(\mathcal{B}\geq u_p\right)=\frac{w}{p}+o(p^{-1}).
\end{align}
In addition, \eqref{B>=up} also holds for any $u'_p=u_p+o\{(\log p)^{-1/2}\}$.
\end{lemma}

\begin{proof}
Let $u'_p=u_p+\xi_p$ for some $\xi_p$ to be specified soon. By Lemma \ref{lem:sup_Bt}, $\pr(\mathcal{B}\geq u'_p)=2\sum_{k=1}^{\infty}(-1)^{k+1}\exp(-2k^2{u'}_p^2)$. It can be verified that $$2\exp(-2{u'}_p^2) = \frac{w}{p}\exp(-4\xi_p u_p)\exp(-2\xi_p^2) = \frac{w}{p}\{1+o(1)\},$$ if $\xi_p=o\{(\log p)^{-1/2}\}$. Further, we can show that $$2\sum_{k=2}^{\infty}(-1)^{k+1}\exp(-2k^2{u'}_p^2)=2\exp(-2{u'}_p^2)\{1+o(1)\},$$ from which the conclusion follows.
\end{proof}

\begin{lemma}\label{lemma:B-Bn}
If $\log p=o(n^{\frac{1}{2+\upsilon}})$ for some $\upsilon>0$, then
\begin{align*}
    0\leq \pr(\mathcal{B}\geq u_p+\xi_p) - \pr(\mathcal{B}_n\geq u_p+\xi_p) = o(p^{-1}),
\end{align*}
for any $\xi_p=o\{(\log p)^{-1/2}\}$.
\end{lemma}

\begin{proof}
Let $y=(\log p)^{-(1+\upsilon)/2}$. By Lemma \ref{lem:sup_Wt}, if $\log p=o(n^{\frac{1}{2+\upsilon}})$, then
\begin{align*}
    \pr\left(\max_{0\leq k\leq n-1}\sup_{\frac{k}{n}\leq t\leq\frac{k+1}{n}}|\mathcal{W}_t-\mathcal{W}_{\frac{k}{n}}|\geq y\right)
    &\leq \sum_{k=0}^{n-1}\pr\left(\sup_{0\leq t\leq\frac{1}{n}}|\mathcal{W}_{t}|\geq y\right)\\
    &= n\pr\left(\sup_{0\leq t\leq 1}|\mathcal{W}_{t}|\geq \sqrt{n}y\right)\\
    &\lesssim n\exp(-ny^2/2) = o(p^{-1}).
\end{align*}
Consequently,
\begin{align*}
    \pr(\mathcal{B}-\mathcal{B}_n \geq y)
    &= \pr\left(\max_{0\leq k\leq n-1}\left(\sup_{\frac{k}{n}\leq t\leq\frac{k+1}{n}}|\mathcal{W}_t-t\mathcal{W}_1|-|\mathcal{W}_{\frac{k}{n}}-\frac{k}{n}\mathcal{W}_1|\right)\geq y\right)\\
    &\leq \pr\left(\max_{0\leq k\leq n-1}\sup_{\frac{k}{n}\leq t\leq\frac{k+1}{n}}|\mathcal{W}_t-\mathcal{W}_{\frac{k}{n}}|\geq \frac{y}{2}\right) + \pr\left(\max_{0\leq k\leq n-1}\sup_{\frac{k}{n}\leq t\leq\frac{k+1}{n}}|(t-\frac{k}{n})\mathcal{W}_1|\geq \frac{y}{2}\right)\\
    &\leq (n+1)\pr\left(\sup_{0\leq t\leq 1}|\mathcal{W}_{t}|\geq \sqrt{n}y/2\right) = o(p^{-1}).
\end{align*}

Let $u'_p=u_p+\xi_p$. Then
\begin{align*}
    \pr(\mathcal{B}\geq u'_p+y)
    &= \pr(\mathcal{B}\geq u'_p+y, \mathcal{B}_n\geq u'_p) + \pr(\mathcal{B}\geq u'_p+y, \mathcal{B}_n\leq u'_p)\\
    &\leq \pr(\mathcal{B}_n\geq u'_p) + \pr(\mathcal{B}-\mathcal{B}_n \geq y)\\
    &\leq \pr(\mathcal{B}_n\geq u'_p) + o(p^{-1}).
\end{align*}
Noticing that $\mathcal{B}_n\leq\mathcal{B}$, we have $\pr(\mathcal{B}_n\geq u'_p)\leq\pr(\mathcal{B}\geq u'_p)$. Combining these facts together and by Lemma \ref{lemma:B>=up}, the conclusion follows.
\end{proof}

\begin{lemma}\label{lemma:Z0j_weak}
Suppose $X_i\sim N(0,\Sigma)$ for $i=1,\ldots,n$, with $\Sigma=(\rho_{j_1j_2})_{p\times p}$. If $|\rho_{j_1j_2}|\leq\delta_p$ for all $1\leq j_1<j_2\leq p$, where $\delta_p$ is a sequence of positive constants such that $\delta_p=o\{(\log p)^{-1}\}$, then for any fixed $m\geq 1$,
\begin{align*}
    \Big(\frac{p}{w}\Big)^m\cdot \pr(Z_{0,j_1}>u_p,\ldots,Z_{0,j_m}>u_p)\to 1,
\end{align*}
uniformly for all $1\leq j_1<\cdots<j_m\leq p$, as $(n,p)\to\infty$.
\end{lemma}

\begin{proof}
We will prove that, for any $u'_p=u_p+o\{(\log p)^{-1/2}\}$,
\begin{align}\label{Z0j1-jm_weak}
    \Big(\frac{p}{w}\Big)^m\cdot \pr(Z_{0,j_1}>u'_p,\ldots,Z_{0,j_m}>u'_p)\to 1,
\end{align}
uniformly for all $1\leq j_1<\cdots<j_m\leq p$, by induction. For $m=1$, by Lemma \ref{lemma:B>=up}, we have $\pr(Z_{0,j}\geq u'_p) = \frac{w}{p}\{1+o(1)\}$.
Assume that \eqref{Z0j1-jm_weak} holds with $m=k-1$ for $k\geq 2$. We claim that it also holds with $m=k$.

The key fact is that $\widetilde{X}_{ij_l}:=(X_{ij_l}-\rho_{j_1j_l}X_{ij_1})/\sqrt{1-\rho_{j_1j_l}^2}\sim N(0,1)$ and is independent of $X_{ij_1}$ for any $l\neq 1$ due to the Gaussianity of $X_{ij}$'s, see Lemma \ref{lem:X2-rX1}.
We observe that, for $l\neq 1$,
\begin{align*}
    Z_{0,j_l} =& \max_{1\leq k<n}\frac{1}{\sqrt{n}}\left|\sum_{i=1}^{k}X_{ij_l}-\frac{k}{n}\sum_{i=1}^{n}X_{ij_l}\right|\\
    \leq& \sqrt{1-\rho_{j_1j_l}^2}\max_{1\leq k<n}\frac{1}{\sqrt{n}}\left|\sum_{i=1}^{k}\widetilde{X}_{ij_l}-\frac{k}{n}\sum_{i=1}^{n}\widetilde{X}_{ij_l}\right|\\
    &+ \max_{1\leq k<n}\frac{1}{\sqrt{n}}\left|\sum_{i=1}^{k}\rho_{j_1j_l}X_{ij_1}-\frac{k}{n}\sum_{i=1}^{n}\rho_{j_1j_l}X_{ij_1}\right|\\
    \leq& \sqrt{1-\rho_{j_1j_l}^2}\widetilde{Z}_{0,j_l} + |\rho_{j_1j_l}|Z_{0,j_1},
\end{align*}
where $\widetilde{Z}_{0,j_l}:=\max_{1\leq k<n}\frac{1}{\sqrt{n}}\left|\sum_{i=1}^{k}\widetilde{X}_{ij_l}-\frac{k}{n}\sum_{i=1}^{n}\widetilde{X}_{ij_l}\right|$, and thus for some $C>0$ that will be specified later,
\begin{align*}
    \pr(Z_{0,j_1}>&u'_p,\ldots,Z_{0,j_k}>u'_p)\\
    \leq& \pr\left(Z_{0,j_1}>u'_p,\widetilde{Z}_{0,j_2}>\frac{u'_p-\rho_{j_1j_2}Z_{0,j_1}}{\sqrt{1-\rho_{j_1j_2}^2}},\ldots,\widetilde{Z}_{0,j_k}>\frac{u'_p-\rho_{j_1j_k}Z_{0,j_1}}{{\sqrt{1-\rho_{j_1j_k}^2}}}\right)\\
    =& \pr\left(Cu'_p\geq Z_{0,j_1}>u'_p,\widetilde{Z}_{0,j_2}>u'_p-\delta_p Z_{0,j_1},\ldots,\widetilde{Z}_{0,j_k}>u'_p-\delta_p Z_{0,j_1}\right)\\
    &+ \pr\left(Z_{0,j_1}>Cu'_p,\widetilde{Z}_{0,j_2}>u'_p-\delta_p Z_{0,j_1},\ldots,\widetilde{Z}_{0,j_k}>u'_p-\delta_p Z_{0,j_1}\right)\\
    \leq& \pr\left(Cu'_p\geq Z_{0,j_1}>u'_p,\widetilde{Z}_{0,j_2}>u'_p-C\delta_p u'_p,\ldots,\widetilde{Z}_{0,j_k}>u'_p-C\delta_p u'_p\right)\\
    &+ \pr(Z_{0,j_1}>Cu'_p)\\
    :=& I+II.
\end{align*}
By the independence between $Z_{0,j_1}$ and $\{\widetilde{Z}_{0,j_2},\ldots,\widetilde{Z}_{0,j_k}\}$, we have
\begin{align*}
    I &= \pr(Cu'_p\geq Z_{0,j_1}>u'_p)\pr(\widetilde{Z}_{0,j_2}>u'_p-C\delta_p u'_p,\ldots,\widetilde{Z}_{0,j_k}>u'_p-C\delta_p u'_p)\\
    &\leq \pr(Z_{0,j_1}>u'_p)\pr(\widetilde{Z}_{0,j_2}>u'_p-C\delta_p u'_p,\ldots,\widetilde{Z}_{0,j_k}>u'_p-C\delta_p u'_p)\\
    &= \Big(\frac{w}{p}\Big)^k\{1+o(1)\},
\end{align*}
where the last equality is followed by induction due to the fact that $\delta_p u'_p=o\{(\log p)^{-1/2}\}$.
Additionally, by using arguments similar to those in the proof of Lemma \ref{lemma:B>=up},
\begin{align*}
    II=\pr(\mathcal{B}>Cu'_p)\{1+o(1)\}=O(p^{-C^2})=o(p^{-k}),
\end{align*}
if $C$ is chosen such that $C>\sqrt{k}$. Hence,
\begin{align*}
    \pr(Z_{0,j_1}>u'_p,\ldots,Z_{0,j_k}>u'_p)\leq\Big(\frac{w}{p}\Big)^k\{1+o(1)\}.
\end{align*}

Similarly, we can show that
\begin{align*}
    \pr(Z_{0,j_1}>u'_p,\ldots,Z_{0,j_k}>u'_p)\geq\Big(\frac{w}{p}\Big)^k\{1+o(1)\}.
\end{align*}
Hence the conclusion follows.
\end{proof}

\begin{lemma}\label{lemma:Z0j_strong}
Suppose $X_i\sim N(0,\Sigma)$ for $i=1,\ldots,n$, with $\Sigma=(\rho_{j_1j_2})_{p\times p}$. For any fixed $m\geq 2$, consider $m$ distinct integers $j_1,\ldots,j_m\in\{1,\ldots,p\}$. If there exist $0<\delta_p<\varrho<1$ satisfying $\delta_p=o\{(\log p)^{-1}\}$ such that $|\rho_{j_1j_2}|\leq\varrho$ and $|\rho_{j_kj_l}|\leq\delta_p$ for all $1\leq k,l\leq m$ but $(k,l)\neq(1,2)$, then
\begin{align}\label{Z0j1-jm_strong}
    \pr(Z_{0,j_1}>u_p,Z_{0,j_2}>u_p,\ldots,Z_{0,j_m}>u_p)\lesssim p^{-(m-\zeta)},
\end{align}
uniformly for all $j_1,\ldots,j_m$, where $\zeta=2-C_0^2\in (0,1)$ and
\begin{align*}
    C_0<\frac{\varrho-\sqrt{2}(1-\varrho^2)}{2\varrho^2-1}.
\end{align*}

\end{lemma}

\begin{proof}
Without loss of generality, we assume that $j_1,j_2$ are always included in $\{j_1,\ldots,j_m\}$ for $m\geq 2$. Otherwise, the conclusion follows by Lemma \ref{lemma:Z0j_weak} immediately.
We first show that \eqref{Z0j1-jm_strong} holds for $m=2$.
By using arguments similar to those in the proof of Lemma \ref{lemma:Z0j_weak}, we have
\begin{align*}
    \pr(Z_{0,j_1}>u'_p,Z_{0,j_2}>u'_p)
    &\leq \pr(Z_{0,j_1}>u'_p)\pr\left(\widetilde{Z}_{0,j_2}>\frac{(1-|\rho_{j_1j_2}C|)u'_p}{\sqrt{1-\rho_{j_1j_2}^2}}\right) + \pr(Z_{0,j_1}>Cu'_p)\\
    &\lesssim p^{-1-\frac{(1-|\rho_{j_1j_2}|C)^2}{1-\rho_{j_1j_2}^2}} + p^{-C^2}.
\end{align*}
By choosing $C=C_0\in(1,\sqrt{2})$ such that $C_0<\frac{\varrho-\sqrt{2}(1-\varrho^2)}{2\varrho^2-1}
$, it can be verified that
\begin{align*}
    \pr(Z_{0,j_1}>u'_p,Z_{0,j_2}>u'_p)
    \lesssim p^{-C_0^2} = p^{-(2-\zeta)}.
\end{align*}

Then we turn to $m=3$. With a little abuse of notations, for any $l,l'\in\{1,\ldots,m\}$ with $l\neq 1,2$ and $l'\neq l$, let $\check{Z}_{0,j_{l'}}=\max_{1\leq k<n}n^{-1/2}\left|\sum_{i=1}^{k}\check{X}_{ij_{l'}}-\frac{k}{n}\sum_{i=1}^{n}\check{X}_{ij_{l'}}\right|$ with
$$\check{X}_{ij_{l'}}:=(X_{ij_{l'}}-\rho_{j_lj_{l'}}X_{ij_l})/\sqrt{1-\rho_{j_lj_{l'}}^2}\sim N(0,1).$$
Similarly, we can show that, with $C_m=\sqrt{m}$,
\begin{align*}
    \pr(Z_{0,j_1}>&u'_p,Z_{0,j_2}>u'_p,Z_{0,j_l}>u'_p)\\
    &\leq \pr(Z_{0,j_l}>u'_p)\pr\left(\check{Z}_{0,j_1}>{u'_p-C_m\delta_pu'_p},\check{Z}_{0,j_2}>{u'_p-C_m\delta_pu'_p}\right) + \pr(Z_{0,j_l}>C_3u'_p)\\
    &\lesssim p^{-1}p^{-\{(m-1)-\zeta\}} + O(p^{-C_m^2})\\
    &\lesssim p^{-(m-\zeta)}.
\end{align*}

Finally, by induction, the conclusion follows.
\end{proof}

\begin{proposition}\label{null:Mnp-Gaussian}
If Assumption \ref{asmp:corr_spatial} holds and $\log p=o(n^{\frac{1}{2+\upsilon}})$ for some $\upsilon>0$, then
\begin{align*}
    \lim_{(n,p)\to\infty}\pr\left(M_{n,p}\leq u_p\{\exp(-x)\}\right) = \exp\{-\exp(-x)\}.
\end{align*}
\end{proposition}

\begin{proof}
\textbf{Step 1}: Recall that $C_p=\{j:|B_{p,j}|\geq p^\kappa\}$, where $B_{p,j}=\{j':|\rho_{jj'}|\geq\delta_p\}$ for $j=1,\ldots,p$ and $\kappa:=\kappa_p\to 0$. By Lemmas \ref{lemma:B>=up}--\ref{lemma:B-Bn} and the assumption that $|C_p|/p\to 0$, we have
\begin{align*}
    \pr(\max_{j\in C_p}Z_{0,j}>u_p)\leq \sum_{j\in C_p}\pr(Z_{0,j}\geq u_p)\to 0.
\end{align*}
Let $D_p=\{j:|B_{p,j}|<p^\kappa\}$ and thus $|D_p|/p\to 1$. Notice that
\begin{align*}
    \pr(\max_{j\in D_p}Z_{0,j}>u_p)
    \leq \pr(\max_{1\leq j\leq p}Z_{0,j}>u_p)
    \leq \pr(\max_{j\in C_p}Z_{0,j}>u_p) + \pr(\max_{j\in D_p}Z_{0,j}>u_p).
\end{align*}
Hence, it suffices to show that
\begin{align}\label{tail:D_p}
    \pr(\max_{j\in D_p}Z_{0,j} > u_p) \to 1-\exp\{-\exp(-x)\}.
\end{align}

\textbf{Step 2}: The key component of the proof is to construct sharp lower and upper bounds of $\pr(\max_{j\in D_p}Z_{0,j} > u_p)$ by using the \textit{inclusion-exclusion principle}. By the inclusion-exclusion principle, for any $k\geq 1$,
\begin{align*}
    \sum_{t=1}^{2k}(-1)^{t-1}\alpha_t
    \leq \pr(\max_{j\in D_p}Z_{0,j}>u_p)
    \leq \sum_{t=1}^{2k+1}(-1)^{t-1}\alpha_t,
\end{align*}
where
\begin{align}\label{alpha_t}
    \alpha_t=\sum \pr(Z_{0,j_1}>u_p,\ldots,Z_{0,j_t}>u_p)
\end{align}
and the sum runs over all combinations $j_1,\ldots,j_t\in D_p$ such that $j_1<\cdots<j_t$. If we can show that
\begin{align}\label{alpha_t_limit}
    \alpha_t\to \frac{1}{t!}\exp(-tx),
\end{align}
then, for any $k\geq 1$,
\begin{align*}
    \sum_{t=1}^{2k}(-1)^{t-1}\frac{\{\exp(-x)\}^t}{t!}
    &\leq \liminf_{(n,p)\to\infty}\pr(\max_{j\in D_p}Z_{0,j}>u_p)\\
    &\leq \limsup_{(n,p)\to\infty}\pr(\max_{j\in D_p}Z_{0,j}>u_p)\\
    &\leq \sum_{t=1}^{2k+1}(-1)^{t-1}\frac{\{\exp(-x)\}^t}{t!}.
\end{align*}
By letting $k\to\infty$ and using the Taylor expansion of the function $1-\exp(-x)$, \eqref{tail:D_p} immediately follows.

\textbf{Step 3}: Evidently, by Lemmas \ref{lemma:B>=up}--\ref{lemma:B-Bn} and the assumption $|D_p|/p\to 1$, \eqref{alpha_t_limit} holds with $t=1$. Next we will show that it holds for any $t\geq 2$.

Let
\begin{align*}
    \big\{(j_1,\ldots,j_t)\in (D_p)^t: j_1<\cdots<j_t\big\}=F_t\bigcup G_t,
\end{align*}
where
\begin{align}
    F_t :=& \big\{(j_1,\ldots,j_t)\in (D_p)^t: j_1<\cdots<j_t\ \mbox{and}\ |\rho_{j_rj_s}|\leq\delta_p\ \mbox{for all}\ 1\leq r<s\leq t\}\ \text{and}\ \nonumber\\
    G_t :=& \big\{(j_1,\ldots,j_t)\in (D_p)^t: j_1<\cdots<j_t\ \mbox{and}\ |\rho_{j_rj_s}|>\delta_p\ \mbox{for some pair}\ (r,s)\ \mbox{with}\ \nonumber\\
    &~~~~~~~~~~~~~~~~~~~~~~~~~~~~~~~~~~~~~~~~~~~~~~~~~~~~~~~~~~~~~~~~~~~~~~~~~~~~~~~1\leq r<s\leq t\big\}.\label{Gt}
\end{align}
Now, think $D_p$ as a graph with $|D_p|$ vertices. Keep in mind that $|D_p|\leq p$ and $|D_p|/p\to 1$. Any two different vertices from them, say, $j$ and $k$, are connected if $|\rho_{jk}|>\delta_p$. In this case we also say there is an edge between them. By the definition of $D_p$, each vertex in the graph has at most $p^{\kappa}$ neighbors.
{Replacing ``$n$'', ``$q$'' and ``$t$'' in Lemma \ref{lem:Feng}--(i) with ``$|D_p|$'', ``$p^{\kappa}$'' and ``$t$'', respectively,} we have $|G_t|\leq p^{t+\kappa-1}$ for each $t\geq 2$. Therefore $\binom{|D_p|}{t}\geq|F_t|\geq\binom{|D_p|}{t}-p^{t+\kappa-1}$. Since $D_p/p\to 1$ and $\kappa\to 0$, we know
\begin{align}\label{No.weak}
    \frac{|F_t|}{p^t}\to\frac{1}{t!}.
\end{align}
Reviewing $\alpha_t$ in \eqref{alpha_t}, we see
\begin{align*}
    \alpha_t =
    \sum_{(j_1,\ldots,j_t)\in F_t}\pr(Z_{0,j_1}>u_p,\ldots,Z_{0,j_t}>u_p) +
    \sum_{(j_1,\ldots,j_t)\in G_t}\pr(Z_{0,j_1}>u_p,\ldots,Z_{0,j_t}>u_p).
\end{align*}
From Lemma \ref{lemma:Z0j_weak} and \eqref{No.weak} we have
\begin{align*}
    \sum_{(j_1,\ldots,j_t)\in F_t}\pr(Z_{0,j_1}>u_p,\ldots,Z_{0,j_t}>u_p)\to\frac{1}{t!}\exp(-tx).
\end{align*}
As a consequence, to derive \eqref{alpha_t_limit}, it remains to show that, for each $t\geq 2$,
\begin{align}\label{sum-Gt}
    \sum_{(j_1,\ldots,j_t)\in G_t}\pr(Z_{0,j_1}>u_p,\ldots,Z_{0,j_t}>u_p)\to 0.
\end{align}

\textbf{Step 4}: If $t=2$, the sum of probabilities in \eqref{sum-Gt} is bounded by $|G_2|\cdot\max_{1\leq j_1<j_2\leq p}\pr(Z_{0,j_1}>u_p, Z_{0,j_2}>u_p)$. By Lemma \ref{lem:Feng}--(i), $|G_2|\leq p^{\kappa+1}$. Then by Lemma \ref{lemma:Z0j_strong}, \eqref{sum-Gt} follows. Hence it remains to show \eqref{sum-Gt} holds for $t\geq 3$.

For a set $A\subset\{1,2,\ldots,m\}$ with $2\leq m\leq p$, define
\begin{align*}
\wp(A)=\max\Big\{|S|: S\subset A\ \mbox{and}\ \max_{j,k\in S, j\neq k}|\rho_{jk}|\leq\delta_p\Big\}.
\end{align*}
Easily, $\wp(A)$ takes possible values $0,2,\ldots,|A|$, where we regard $|\varnothing|=0$. If $\wp(A)=0$, then $|\rho_{jk}|>\delta_p$ for all $j,k\in A$. Compare it with $G_t$ from \eqref{Gt}. We will look at $G_t$ closely. To do so, we classify $G_t$ into the following subets
\begin{align*}
    G_{t,j}=\big\{(j_1,\ldots,,j_t)\in G_t: \wp(\{j_1,\ldots,,j_t\})=j\big\},
\end{align*}
for $j=0,2,\ldots,t-1$. By the definition of $G_t$, we see $G_t=\bigcup G_{t,j}$. Since $t\geq 3$ is fixed, to show \eqref{sum-Gt}, it suffices to prove
\begin{align}\label{sum-Gtj}
    \sum_{(j_1,\ldots,,j_t)\in G_{t,j}}\pr(Z_{0,j_1}>u_p,\ldots,Z_{0,j_t}>u_p)\to 0,
\end{align}
for any $j\in\{0,2,\ldots,t-1\}$.

\textbf{Step 5}: Assume $(j_1,\ldots,j_t)\in G_{t,0}$. This implies that $|\rho_{j_rj_s}|>\delta_p$ for all $1\leq r<s\leq t$. Therefore, the subgraph $\{j_1,\ldots,,j_t\}\in  G_t$ is a clique. Replacing ``$n$'', ``$q$'' and ``$t$'' in Lemma \ref{lem:Feng}--(ii) with ``$|D_p|$'', ``$p^{\kappa}$'' and ``$t$'', respectively, we have $|G_{t,0}|\leq p^{1+\kappa(t-1)}\leq p^{1+t\kappa}$. Together with Lemma \ref{lemma:Z0j_strong}, the sum of probabilities in \eqref{sum-Gtj} is bounded by
\begin{align*}
    p^{1+t\kappa}\cdot\max_{1\leq j_1<j_2\leq p}\pr(Z_{0,j_1}>u_p,Z_{0,j_2}>u_p)\to 0.
\end{align*}
In other words, \eqref{sum-Gtj} holds with $j=0$.

Now we assume $(j_1,\ldots,j_t)\in G_{t,j}$ with $j\in\{2,\ldots,t-1\}$. By definition, there exists $S\subset\{j_1,\ldots,j_t\}$ such that $\max_{j_1,j_2\in S, j_1\neq j_2}|\rho_{j_1j_2}|\leq\delta_p$ and for each $k\in\{j_1,\ldots,j_t\}\backslash S$, there exists $l\in S$ satisfying $|\rho_{kl}|>\delta_p$. Looking at the last statement we see two possibilities: (i) for each $k\in\{j_1,\ldots,j_t\}\backslash S$, there exist at least two indices, say, $l,l'\in S, l\neq l'$ satisfying $|\rho_{kl}|>\delta_p$ and $|\rho_{kl'}|>\delta_p$; (ii) there exists $k\in\{j_1,\ldots,j_t\}\backslash S$ for which $|\rho_{kl}|>\delta_p$ for an unique $l\in S$. However, for $(j_1,\ldots,j_t)\in G_{t,j}$, (i) and (ii) could happen at the same time for different $S$, say, (i) holds for $S_1$ and (ii) holds for $S_2$ simultaneously. Thus, to differentiate the two cases, we introduce following two definitions. Set
\begin{align}\label{Htj}
    H_{t,j} = \big\{&(j_1,\ldots,j_t)\in G_{t,j}: \mbox{there exists an}\ S\subset\{j_1,\ldots,j_t\}\ \mbox{with}\ |S|=j\ \mbox{and}\nonumber\\
    &\max_{j_1,j_2\in S,j_1\neq j_2}|\rho_{j_1j_2}|\leq\delta_p\ \mbox{such that for any}\ k\in\{j_1,\ldots,j_t\}\backslash S\ \mbox{there exist}\nonumber\\
    &l,l'\in S,l\neq l'\ \mbox{satisfying}\ \min\{|\rho_{kl}|,|\rho_{kl'}|\}>\delta_p\big\}.
\end{align}
By Lemma \ref{lem:Feng}--(iii), $|H_{t,j}|\leq t^t\cdot p^{j-1+(t-j+1)\kappa}$ for each $t\geq 3$.
On the other hand, set
\begin{align*}
    H_{t,j}' = \big\{&(j_1,\ldots,j_t)\in G_{t,j}: \mbox{for any}\ S\subset\{j_1,\ldots,j_t\}\ \mbox{with}\ |S|=j\ \mbox{and}\nonumber\\
    &\max_{j_1,j_2\in S,j_1\neq j_2}|\rho_{j_1j_2}|\leq\delta_p,\ \mbox{there exists}\ k\in\{j_1,\ldots,j_t\}\backslash S\ \mbox{such that}\ |\rho_{kl}|>\delta_p\nonumber\\
    &\mbox{for a unique}\ l\in S\}.
\end{align*}
From Lemma \ref{lem:Feng}--(iv), we see $|H_{t,j}'|\leq t^t\cdot p^{j+(t-j)\kappa}$.
Due to the fact that $G_{t,j}=H_{t,j}\bigcup H_{t,j}'$, to show \eqref{sum-Gtj}, it suffices to show
\begin{align}\label{sum-Htj}
    \sum_{(j_1,\ldots,j_t)\in H_{t,j}}\pr(Z_{0,j_1}>u_p,\ldots,Z_{0,j_t}>u_p)\to 0,
\end{align}
and
\begin{align}\label{sum-H'tj}
    \sum_{(j_1,\ldots,j_t)\in H_{t,j}'}\pr(Z_{0,j_1}>u_p,\ldots,Z_{0,j_t}>u_p)\to 0,
\end{align}
for $j=2,\ldots,t-1$.

\textbf{Step 6}: Consider the $S$ in \eqref{Htj}. By Lemma \ref{lemma:Z0j_weak},
\begin{align*}
    \sum_{(j_1,\ldots,j_t)\in H_{t,j}}\pr(Z_{0,j_1}>u_p,\ldots,Z_{0,j_t}>u_p)
    &\leq \sum_{(j_1,\ldots,j_t)\in H_{t,j}}\pr(\bigcap_{k\in S}\{Z_{0,k}>u_p\})\\
    &\lesssim t^t\cdot p^{j-1+(t-j+1)\kappa}\cdot p^{-j}\\
    &\lesssim t^t\cdot p^{-1+t\kappa}.
\end{align*}
Consequently, we obtain \eqref{sum-Htj}, since $\kappa\to 0$.

We then turn to \eqref{sum-H'tj}. For $(j_1,\ldots,j_t)\in H_{t,j}'$, pick $S\subset\{j_1,\ldots,j_t\}$ with $|S|=j$ and $\max_{j_1,j_2\in S,j_1\neq j_2}|\rho_{j_1j_2}|\leq\delta_p$, and $k\in\{j_1,\ldots,j_t\}\backslash S$ such that $|\rho_{kl}|>\delta_p$ for a unique $l\in S$. By Lemma \ref{lemma:Z0j_strong},
\begin{align*}
    \sum_{(j_1,\ldots,j_t)\in H_{t,j}'}\pr(Z_{0,j_1}>u_p,\ldots,Z_{0,j_t}>u_p)
    &\leq \sum_{(j_1,\ldots,j_t)\in H_{t,j}'}\pr(Z_{0,k}\bigcap_{k'\in S}\{Z_{0,k'}>u_p\})\\
    &\lesssim t^t\cdot p^{j+(t-j)\kappa}\cdot p^{-(j+1-\zeta)}\\
    &\lesssim t^t\cdot p^{-(1-\zeta)+t\kappa}.
\end{align*}
Consequently, we obtain \eqref{sum-H'tj}, since $\kappa\to 0$.

Hence Proposition \ref{null:Mnp-Gaussian} follows.
\end{proof}

\subsubsection{Asymptotic null distribution of $M^{\dagger}_{n,p}$}

Recall that $A(x)=\sqrt{2\log x}$ and $D(x)=2\log x+2^{-1}\log\log x-2^{-1}\log\pi$. For any fixed $x$, let $w=\exp(-x)$ and $v_{p,n}:=\{x+D(p\log h_n)\}/A(p\log h_n)$, where $h_n=\left\{(\lambda_n/n)^{-1}-1\right\}^2$.
Let $\mathcal{M}(\lambda_n)=\sup_{\lambda_n/n\leq t\leq 1-\lambda_n/n}\{t(1-t)\}^{-1/2}|\mathcal{W}_t-t\mathcal{W}_1|$ and $\mathcal{M}_n(\lambda_n)=\sup_{\lambda_n\leq k\leq n-\lambda_n}\{(k/n)(1-k/n)\}^{-1/2}|\mathcal{W}_{k/n}-(k/n)\mathcal{W}_1|$.
Let $Z_{0.5,j}=\max_{\lambda_n\leq k<n-\lambda_n}|C_{0.5,j}(k)|$ for $j=1,\ldots,p$.
Let $\mathcal{V}$ be an Ornstein-Uhlenbeck process with $\E\big(\mathcal{V}(t)\big)=0$ and $\E\big(\mathcal{V}(t)\mathcal{V}(s)\big)=\exp(-|t-s|/2)$.
{Since $\lambda_n\sim n^{\epsilon}$ with $\epsilon\in(0,1)$, we have $h_n\to\infty$.}

\begin{lemma}\label{lemma:M>=vpn}
It holds that
\begin{align}\label{M>=vpn}
    \pr\left(\mathcal{M}(\lambda_n)\geq v_{p,n}\right)=\frac{w}{p}+o(p^{-1}).
\end{align}
In addition, \eqref{M>=vpn} also holds for any $v'_{p,n}=v_{p,n}+o\big(\{\log(p\log h_n)\}^{-1/2}\big)$.
\end{lemma}

\begin{proof}
Let $v'_{p,n}=v_{p,n}+\eta_{p,n}$ for some $\eta_{p,n}$ to be specified soon. Since $h_n\to\infty$, according to (A.3.19) in \cite{csorgo1997limit} and by Lemma \ref{lem:V(t)}, we have
\begin{align*}
    \pr\left(\mathcal{M}(\lambda_n)\geq v'_{p,n}\right)
    &= \pr\left(\sup_{0\leq t\leq\log h_n}|\mathcal{V}(t)|\geq v'_{p,n}\right)\\
    &= \frac{v'_{p,n}\exp(-{v'}_{p,n}^2/2)}{\sqrt{2\pi}}\log h_n\{1+o(1)\}.
\end{align*}
By choosing $\eta_{p,n}=o\big(\{\log(p\log h_n)\}^{-1/2}\big)$, it can be verified that
\begin{align*}
    \frac{v'_{p,n}\exp(-{v'}_{p,n}^2/2)}{\sqrt{2\pi}}\log h_n
    &= \exp\left\{-x-\log p-\sqrt{2\log(p\log h_n)}\xi_{p,n}+o(1)\right\}\\
    &= \frac{\exp(-x)}{p}\{1+o(1)\}.
\end{align*}
Hence the conclusion follows.
\end{proof}

\begin{lemma}\label{lemma:M-Mn}
If $\log p=o(n^{\frac{\epsilon}{2+\upsilon}})$ for some $\upsilon>0$, then
\begin{align*}
    0\leq \pr(\mathcal{M}(\lambda_n)\geq v_{p,n}+\eta_{p,n}) - \pr(\mathcal{M}_n(\lambda_n)\geq v_{p,n}+\eta_{p,n}) = o(p^{-1}),
\end{align*}
for any $\eta_{p,n}=o\big(\{\log(p\log h_n)\}^{-1/2}\big)$.
\end{lemma}

\begin{proof}
Let $y=\{\log (p\log h_n)\}^{-(1+\upsilon)/2}$. Notice that $\inf_{\lambda_n/n\leq t\leq 1-\lambda_n/n}t(1-t) = \lambda_n/n(1-\lambda_n/n)$. By using arguments similar to those in the proof of Lemma \ref{lemma:B-Bn}, we have
\begin{align*}
    \pr&\left(\max_{\lambda_n\leq k\leq n-\lambda_n}\sup_{\frac{k}{n}\leq t\leq\frac{k+1}{n}}\frac{|\mathcal{W}_t-\mathcal{W}_{\frac{k}{n}}|}{\sqrt{t(1-t)}}\geq y\right)\\
    &\leq \pr\left({
    \max_{\lambda_n\leq k\leq n-\lambda_n}\sup_{\frac{k}{n}\leq t\leq\frac{k+1}{n}}{|\mathcal{W}_t-\mathcal{W}_{\frac{k}{n}}|}}\geq y{\sqrt{\lambda_n/n(1-\lambda_n/n)}}\right)\\
    &\lesssim n\exp\{-n\lambda_n/n(1-\lambda_n/n)y^2/2\} = o(p^{-1}),
\end{align*}
if $\log p=o(n^{\frac{\epsilon}{2+\upsilon}})$.
Consequently, by using arguments similar to those in the proof of Lemma \ref{lemma:B-Bn},
\begin{align*}
    \pr\left(\mathcal{M}(\lambda_n) -
    \max_{\lambda_n\leq k\leq n-\lambda_n}\sup_{\frac{k}{n}\leq t\leq\frac{k+1}{n}}\frac{|\mathcal{W}_{\frac{k}{n}}-\frac{k}{n}\mathcal{W}_1|}{\sqrt{t(1-t)}} \geq y
    \right) = o(p^{-1}).
\end{align*}

Let $v'_{p,n}=v_{p,n}+\eta_{p,n}$. Again, by using arguments similar to those in the proof of Lemma \ref{lemma:B-Bn}, the conclusion follows.
\end{proof}

\begin{proposition}\label{null:MMnp-Gaussian}
If Assumption \ref{asmp:corr_spatial} holds and $\log p=o(n^{\frac{\epsilon}{2+\upsilon}})$ for some $\upsilon>0$, then
\begin{align*}
    \lim_{(n,p)\to\infty} \pr\left(M^{\dagger}_{n,p}\leq \frac{x+D(p\log h_n)}{A(p\log h_n)}\right) = \exp\{-\exp(-x)\},.
\end{align*}
\end{proposition}

\begin{proof}
By using arguments similar to those in the proof of Proposition \ref{null:Mnp-Gaussian}, the conclusion follows.
\end{proof}

\subsection{For non-Gaussian observation under Assumption \ref{asmp:corr_temporal}}\label{M-non-Gaussian}

\subsubsection{Asymptotic null distribution of $M_{n,p}$}

By Lemma E.6 in \cite{10.1214/15-aos1347}, we can select $\widehat{\sigma}_j$'s such that
\begin{align*}
    \pr\left(\max_{1\leq j\leq p}\left|\widehat{\sigma}_{j}-\sigma_{j}\right|\geq (2\log p)^{-4}\right)\lesssim n^{-C},\ C>0.
\end{align*}
By using arguments similar to those in Step 4 in the proof of Theorem C.4 in \cite{10.1214/15-aos1347}, it suffices to show that
\[
    \lim _{(n,p)\to\infty} \pr\left(\max_{1\leq j\leq p}\max_{1\leq k\leq n}|C_{0,j}(k)|/\sigma_{j}\leq u_{p}\{\exp(-x)\}\right)
    = \exp\{-\exp(-x)\}.
\]
According to Steps 1--3 in the proof of Theorem C.4 in \cite{10.1214/15-aos1347}, we have
\begin{align*}
    \left|\pr\left(\max_{1\leq j\leq p}\max_{1\leq k\leq n}|C_{0,j}(k)|/\sigma_{j}\leq u_{p}\{\exp(-x)\}\right)
    -\pr\left(\max_{1\leq j\leq p}\max_{1\leq k\leq n}|\widetilde{C}_{0,j}(k)|/\sigma_{j}\leq u_{p}\{\exp(-x)\}\right)\right| = o(1),
\end{align*}
where
\begin{align*}
    \widetilde{C}_{0,j}(k)=n^{-1/2}\left(\sum_{i=1}^{k}Y_{ij}-\frac{k}{n}\sum_{i=1}^{n}Y_{ij}\right),
\end{align*}
and $Y_{i}=(Y_{i1},\ldots,Y_{ip})^\T\sim N(0,\Sigma)$ for $i=1,\cdots,n$.
Consequently, the conclusion follows.

\subsubsection{Asymptotic null distribution of $M^{\dagger}_{n,p}$}

Again, it suffices to show that
\[
    \lim _{(n,p)\to\infty} \pr\left(\max_{1\leq j\leq p}\max_{\lambda_n\leq k\leq n-\lambda_n}|C_{0.5,j}(k)|/\sigma_{j}\leq v_{p,n}\right)
    = \exp\{-\exp(-x)\},
\]
where $v_{p,n}=\frac{x+D(p\log h_n)}{A(p\log h_n)}$.

Define $\mathfrak{t}_{n}=tn$ for $0<t<0.5$. Let $g_{p,n}=\sqrt{2\log(p\log h_n)}$. Obviously, $v_{p,n}\sim g_{p,n}$. Notice that
\begin{align*}
    &\pr\left(\max_{\lambda_n\leq k\leq n-\lambda_n}\left|\sqrt{\frac{k(n-k)}{n}}\left(\frac{1}{k}\sum_{i=1}^{k}X_{ij}-\frac{1}{n-k}\sum_{i=k+1}^{n} X_{ij}\right)\right|\geq \sigma_{j}g_{p,n}/2\right)\\
    &\leq \pr\left(\max_{\lambda_n\leq k\leq n-\lambda_n}|\sum_{i=1}^{k}X_{ij}|\geq\sqrt{\frac{nk}{n-k}}\sigma_{j}g_{p,n}/4\right) +
    \pr\left(\max_{\lambda_n\leq k\leq n-\lambda_n}|\sum_{i=k+1}^{n}X_{ij}|\geq\sqrt{\frac{n(n-k)}{k}}\sigma_{j}g_{p,n}/4\right)\\
    &\leq \pr\left(\max_{\lambda_n\leq k\leq n-\lambda_n}|\sum_{i=1}^{k}X_{ij}|\geq\sqrt{\frac{n\lambda_n}{n-\lambda_n}}\sigma_{j}g_{p,n}/4\right) +
    \pr\left(\max_{\lambda_n\leq k\leq n-\lambda_n}|\sum_{i=k+1}^{n}X_{ij}|\geq\sqrt{\frac{n(n-\mathfrak{t}_{n})}{\mathfrak{t}_{n}}}\sigma_{j}g_{p,n}/4\right).
\end{align*}
By Lemma E.3 in \cite{10.1214/15-aos1347}, we have
\begin{align*}
    \pr\left(\max_{\lambda_n\leq k\leq \mathfrak{t}_{n}}|\sum_{i=1}^{k}X_{ij}|\geq\sqrt{\frac{n\lambda_n}{n-\lambda_n}}\sigma_{j}g_{p,n}/4\right)
    \lesssim \frac{\mathfrak{t}_{n}-\lambda_n}{\left(\sqrt{\frac{n\lambda_n}{n-\lambda_n}}\underline{\sigma}g_{p.n}\right)^q}
    \lesssim \frac{n}{\lambda_n^{q/2}g_{p,n}^q},
\end{align*}
where we recall $\underline{\sigma}=\inf_{1\leq j\leq p}\sigma_{j}$. Similarly,
\begin{align*}
    \pr\left(\max_{\lambda_n\leq k\leq \mathfrak{t}_{n}}|\sum_{i=k+1}^{n}X_{ij}|\geq\sqrt{\frac{n(n-\mathfrak{t}_{n})}{\mathfrak{t}_{n}}}\sigma_{j}g_{p,n}/4\right)
    \lesssim \frac{n}{n^{q/2}g_{p,n}^q}.
\end{align*}
Hence, we conclude that
\begin{align}\label{trun}
    \pr\left(\max_{1\leq j\leq p}\max_{\lambda_n\leq k\leq \mathfrak{t}_{n}}|C_{0.5,j}(k)|/\sigma_{j}\leq g_{p,n}/2\right)
    \lesssim \frac{p}{n^{q\epsilon/2-1}g_{p,n}^q}\to 0,
\end{align}
if $p=o(n^{q\epsilon/2-1})$.
By using similar arguments, we obtain that
\begin{align*}
    \pr\left(\max_{1\leq j\leq p}\max_{n-\mathfrak{t}_{n}\leq k\leq n-\lambda_n}|C_{0.5,j}(k)|/\sigma_{j}\leq g_{p,n}/2\right)
    \lesssim \frac{p}{n^{q\epsilon/2-1}g_{p,n}^q}\to 0.
\end{align*}

Next, for $t\in[0,1]$ and $1\leq j\leq p$, let
\[
    C_{j}^{\sigma}(t) = \frac{1}{\sqrt{n}\sigma_{j}}\left(\sqrt{\frac{n-\lceil nt\rceil}{\lceil nt\rceil}}\sum_{i=1}^{\lceil nt\rceil}X_{ij}-\sqrt{\frac{\lceil nt\rceil}{n-\lceil nt\rceil}}\sum_{i=1}^{n}X_{ij}\right).
\]
Put $t_{n}^{-}=\left\lfloor\mathfrak{t}_{n}\right\rfloor/n$, $t_{n}^{+}=1-\left\lfloor\mathfrak{t}_{n}\right\rfloor/n$ and $t_{i}=i/n$. Then we denote with
\begin{align*}
v_{n, j}(\mathfrak{t}) &= \left(C_{j}^{\sigma}\left(t_{n}^{-}\right), -C_{j}^{\sigma}\left(t_{n}^{-}\right), C_{j}^{\sigma}\left(t_{n}^{-}+t_{1}\right), -C_{j}^{\sigma}\left(t_{n}^{-}+t_{1}\right), \ldots, -C_{j}^{\sigma}\left(t_{n}^{+}\right)\right)^\T\\
v_{n, j}^{-}(\mathfrak{t}) &= \left(C_{j}^{\sigma}\left(t_{[n^\epsilon]}\right), -C_{j}^{\sigma}\left(t_{[n^\epsilon]}\right), C_{j}^{\sigma}\left(t_{[n^\epsilon]+1}\right), -C_{j}^{\sigma}\left(t_{[n^\epsilon]+1}\right), \ldots, -C_{j}^{\sigma}\left(t_{n}^{-}-t_{1}\right)\right)^\T \\
v_{n, j}^{+}(\mathfrak{t}) &= \left(C_{j}^{\sigma}\left(t_{n}^{+}+t_{1}\right), -C_{j}^{\sigma}\left(t_{n}^{+}+t_{1}\right), C_{j}^{\sigma}\left(t_{n}^{+}+t_{2}\right), -C_{j}^{\sigma}\left(t_{n}^{+}+t_{2}\right), \ldots, -C_{j}^{\sigma}(1-t_{[n^\epsilon]})\right)^\T.
\end{align*}
For $1\leq j\leq p$, we stack these vectors together in single vectors
\[
    v_{n}(\mathfrak{t}) = \left(v_{n, 1}(\mathfrak{t}),\ldots,v_{n, p}(\mathfrak{t})\right)^\T,\quad
    v_{n}^{-}(\mathfrak{t}) = \left(v_{n, 1}^{-}(\mathfrak{t}),\ldots,v_{n, p}^{-}(\mathfrak{t})\right)^\T,\quad
    v_{n}^{+}(\mathfrak{t}) = \left(v_{n, 1}^{+}(\mathfrak{t}),\ldots,v_{n, p}^{+}(\mathfrak{t})\right)^\T.
\]
For each of these vectors, we write $\max v_{n}(\mathfrak{t})$ to denote the overall maximum. Note that the dimension of each of these vectors is bounded by $2pn$. Since for each $1\leq i\leq n$
\[
    \sup _{i/n\leq t\leq (i+1)/n}\left|C_{j}^{\sigma}(t)\right|
    =\max\left\{\left|C_{j}^{\sigma}(i/n)\right|, \left|C_{j}^{\sigma}(i/n+1/n)\right|\right\},
\]
we have
\[
    \pr\left(\max v_{n}^{-}(\mathfrak{t}), \max v_{n}(\mathfrak{t}), \max v_{n}^{+}(\mathfrak{t}) \leq v_{p,n}\right)
    =\pr\left(\max_{1\leq j\leq p}\max_{\lambda_n\leq k\leq n-\lambda_n}|C_{0.5,j}(k)|/\sigma_{j} \leq v_{p,n}\right).
\]
Next, observe that for any fixed $x$ we have $v_{p,n}\geq g_{p,n}/2$ for large enough $p$. It follows that for small enough $\mathfrak{t}=\mathfrak{t}_{0}>0$, we have for any fixed $x$
\[
    \pr\left(\max v_{n}^{-}\left(\mathfrak{t}_{0}\right), \max v_{n}^{+}\left(\mathfrak{t}_{0}\right) \leq v_{p,n}\right)=1-o(1).
\]
Hence in order to prove Theorem \ref{null:Max}, it suffices to consider the probability
$$
    \pr\left(\max v_{n}\left(\mathfrak{t}_{0}\right) \leq v_{p,n}\right).
$$

Define
\begin{align*}
    \widetilde{C}_{0.5,j}(k)=\sqrt{\frac{k(n-k)}{n}}\left(\frac{1}{k}\sum_{i=1}^{k}Y_{ij}-\frac{1}{n-k}\sum_{i=1}^{n}Y_{ij}\right),
\end{align*}
where $Y_i=(Y_{i1},\ldots,Y_{ip})^\T\sim N(0,\Sigma)$ independently.
By using similar arguments to those in Step 2 in the proof of Theorem C.4 in \cite{10.1214/15-aos1347}, we have
$$
    \left|\pr\left(\max v_{n}\left(\mathfrak{t}_{0}\right) \leq v_{p,n}\right)-\pr\left(\max_{1\leq j\leq p}\max_{t_n\leq k\leq n-t_n}|\widetilde{C}_{0.5,j}(k)|/\sigma_{j}\leq v_{p,n}\right)\right|
    \lesssim n^{-C} \to 0,
$$
for some positive constant $C>0$. Similar to \eqref{trun}, we can show
$$
    \left|\pr\left(\max_{1\leq j\leq p}\max_{\lambda_n\leq k\leq n-\lambda_n}|\widetilde{C}_{0.5,j}(k)|/\sigma_{j} \leq v_{p,n}\right) -
    \pr\left(\max_{1\leq j\leq p}\max_{t_n\leq k\leq n-t_n} |\widetilde{C}_{0.5,j}(k)|/\sigma_{j} \leq v_{p,n}\right)\right|\to 0.
$$
Consequently, the conclusion follows.

\section{Proof of Theorem \ref{null:DMS}}

\subsection{For Gaussian data sequences}\label{indep:Gaussian}

We first investigate the asymptotic independence of $M_{n,p}$ and $\{S_{n,p}-(n+2)p\}/V_{p,n}^{1/2}$ under $H_0$ if $X_i\sim N(0,\Sigma)$.
Let $\xi_i=\Sigma^{1/2}\varepsilon_i$ for $i=1,\ldots,n$. According to the proof of Theorem 1 in \cite{10.1142/s201032631950014x}, we have
\begin{align*}
    S_{n,p} &\ = \sum_{i=1}^n a_i\xi_i^\T\xi_i+\sum_{i=1}^{n-1}\sum_{k=2}^n b_{i,k}\xi_i^\T\xi_k+o_p\Big(\sqrt{\var(S_{n,p})}\Big)\\
    &:= S^*_{n,p}+o_p\Big(\sqrt{\var(S_{n,p})}\Big),
\end{align*}
where $a_i=\sum_{l=1}^{i-1}(n-l)^{-1}+\sum_{l=i}^{n-1}l^{-1}+n^{-1}-1$ and $b_{i,k}=2\{\sum_{l=1}^{i-1}(n-l)^{-1}+\sum_{l=k}^{n-1}l^{-1}+n^{-1}-1\}$.
Notice that $V_{n,p}$ is a consistent estimator of $\var(S_{n,p})$, see Proposition 1 in \cite{10.1142/s201032631950014x}.
In addition, $\max_{1\leq j\leq p}|\widehat{\sigma}_j-\sigma_j|=O(\sqrt{\log p/n})$ (see Lemma A.2 in \cite{10.1142/s201032631950014x}) and thus
\begin{align*}
    |\max_{1\leq j\leq p}Z_{0,j}/\widehat{\sigma}_j - \max_{1\leq j\leq p}Z_{0,j}/\sigma_j|
    \leq \max_{1\leq j\leq p}Z_{0,j}/\sigma_j \max_{1\leq j\leq p}|\sigma_j/\widehat{\sigma}_j-1|
    = O_p\left(\frac{\log p}{\sqrt{n}}\right)\to 0,
\end{align*}
since $\log p=o(n^{\frac{1}{2+\upsilon}})$.
By Lemma \ref{lem:asy_indep}, it suffices to show that $\{S^*_{n,p}-(n+2)p\}/\nu_{np}$ and $\max_{1\leq j\leq p} \max_{1\leq k\leq n}|C_{0,j}(k)|/\sigma_j$ are asymptotically independent, where $\nu_{n,p}=\sqrt{\var(S_{n,p})}$.

For any fixed $x,y\in\mathbb{R}$, define $A_p:=A_p(x)=\Big\{\frac{S^*_{n,p}-(n+2)p}{\nu_{n,p}}\leq x\Big\}$ and $B_j:=B_j(y)=\{Z_{0,j}>u_p\{\exp(-y)\}\}$ for $j=1,\ldots,p$.
According to Theorem 1 in \cite{10.1142/s201032631950014x}, $\pr(A_p)\to \Phi(x)$,
and by Proposition \ref{null:Mnp-Gaussian}, $\pr(\bigcup_{j=1}^p B_j)\to 1-\exp\{-\exp(-y)\}$.
The goal is to prove that
\begin{align*}
    \pr\left(\frac{S^*_{n,p}-(n+2)p}{\nu_{n,p}}\leq x, \max_{1\leq j\leq p} \max_{1\leq k\leq n}|C_{0,j}(k)|/\sigma_j\leq u_p\{\exp(-y)\}\right)\to\Phi(x)\cdot\exp\{-\exp(-y)\},
\end{align*}
or, equivalently,
\begin{align*}
    \pr\left(\bigcup_{j=1}^p A_pB_j\right)\to\Phi(x)\cdot\{1-\exp\{-\exp(-y)\}\}.
\end{align*}

Let for each $d\geq 1$,
\begin{align*}
    \zeta(p,d) &:= \sum_{1\leq j_1<\cdots<j_d\leq p}\big|\pr(A_p B_{j_1}\cdots B_{j_d}) - \pr(A_p)\cdot\pr(B_{j_1}\cdots B_{j_d})\big|
\end{align*}
and
\begin{align*}
    H(p,d) &:= \sum_{1\leq j_1<\cdots<j_d\leq p}\pr(B_{j_1}\cdots B_{j_d}).
\end{align*}
By the \textit{inclusion-exclusion principle}, we observe that, for any integer $k\geq 1$,
\begin{align*}
    \pr\Big(\bigcup_{j=1}^p A_p B_j\Big)
    \leq \sum_{1\leq j_1\leq p}&\pr(A_pB_{j_1}) - \sum_{1\leq j_1<j_2\leq p}\pr(A_pB_{j_1}B_{j_2}) + \cdots\\
    & + \sum_{1\leq j_1<\cdots<j_{2k+1}\leq p}\pr(A_pB_{j_1}\cdots B_{j_{2k+1}})
\end{align*}
and
\begin{align*}
    \pr\Big(\bigcup_{j=1}^p B_j\Big)
    \geq \sum_{1\leq j_1\leq p}&\pr(B_{j_1}) - \sum_{1\leq j_1<j_2\leq p}\pr(B_{j_1}B_{j_2}) + \cdots\\
    & - \sum_{1\leq j_1<\cdots<j_{2k}\leq p}\pr(B_{j_1}\cdots B_{j_{2k}}).
\end{align*}
Then,
\begin{eqnarray*}
    \pr\Big(\bigcup_{j=1}^p A_p B_j\Big)
    &\leq& \pr(A_p)\Big\{\sum_{1\leq j_1\leq p}\pr(B_{j_1}) - \sum_{1\leq j_1<j_2\leq p}\pr(B_{j_1}B_{j_2}) + \cdots\\
    && - \sum_{1\leq j_1<\cdots<j_{2k}\leq p}\pr(B_{j_1}\cdots B_{j_{2k}})\Big\} + \sum_{d=1}^{2k} \zeta(p,d) + H(p,2k+1)\\
    &\leq& \pr(A_p)\pr\Big(\bigcup_{j=1}^p B_j\Big) + \sum_{d=1}^{2k} \zeta(p,d) + H(p,2k+1).
\end{eqnarray*}
By fixing $k$ and letting $p\to\infty$, we obtain
\begin{align*}
    \limsup_{p\to\infty}\pr\left(\bigcup_{j=1}^p A_pB_j\right)
    \leq \Phi(x)\cdot\{1-\exp\{-\exp(-y)\}\} + \lim_{p\to\infty} H(p,2k+1),
\end{align*}
due to Lemma \ref{lem:indep:A-B}.
According to \eqref{alpha_t} and \eqref{alpha_t_limit}, for each $d\geq 1$, $\lim_{p\to\infty} H(p,d) = \frac{1}{d!}\exp(-dx/2)$.
By letting $k\to\infty$,
\begin{align*}
    \limsup_{p\to\infty}\pr\left(\bigcup_{j=1}^p A_pB_j\right)
    \leq \Phi(x)\cdot\{1-\exp\{-\exp(-y)\}\}.
\end{align*}

By using similar arguments, together with the following results from the {inclusion-exclusion principle}, i.e.,
\begin{align*}
    \pr\left(\bigcup_{j=1}^pA_pB_j\right)
    \geq \sum_{1\leq j_1\leq p}\pr(A_pB_{j_1}) &- \sum_{1\leq j_1<j_2\leq p}\pr(A_pB_{j_1}B_{j_2}) + \cdots\nonumber\\
    & - \sum_{1\leq j_1<\cdots<j_{2k}\leq p}\pr(A_pB_{j_1}\cdots B_{j_{2k}}),
\end{align*}
and
\begin{align*}
    \pr\Big(\bigcup_{j=1}^p B_j\Big)
    \leq \sum_{1\leq j_1\leq p}&\pr(B_{j_1}) - \sum_{1\leq j_1<j_2\leq p}\pr(B_{j_1}B_{j_2}) + \cdots\\
    & + \sum_{1\leq j_1<\cdots<j_{2k-1}\leq p}\pr(B_{j_1}\cdots B_{j_{2k-1}}).
\end{align*}
we obtain
\begin{align*}
    \liminf_{p\to\infty}\pr\left(\bigcup_{j=1}^p A_pB_j\right)
    \geq \Phi(x)\cdot\{1-\exp\{-\exp(-y)\}\}.
\end{align*}
Hence the conclusion follows.

\begin{lemma}\label{lem:indep:A-B}
If the conditions in Theorem \ref{null:DMS} hold, then for each $d\geq 1$, $\zeta(p,d)\to 0$.
\end{lemma}

\begin{proof}
For each $i=1,\ldots,n$, let $\xi_{i,(1)}=(\xi_{i,j_1},\ldots,\xi_{i,j_d})^\T$ and $\xi_{i,(2)}=(\xi_{i,j_{d+1}},\ldots,\xi_{i,j_p})^\T$, and $R_{kl}=\cov(\xi_{i,(k)}, \xi_{i,(l)})$ for $k,l\in\{1,2\}$. By Lemma \ref{lem:X2-rX1}, $\xi_{i,(2)}$ can be decomposed as $\xi_{i,(2)}=U_i+V_i$, where $U_i:=\xi_{i,(2)}-R_{21}R_{11}^{-1}\xi_{i,(1)}$ and $V_i:=R_{21}R_{11}^{-1}\xi_{i,(1)}$ satisfying that $U_i\sim N(0, R_{22}-R_{21}R_{11}^{-1}R_{12})$, $V_i\sim N(0, R_{21}R_{11}^{-1}R_{12})$ and
\begin{align*}
    U_i\ \text{and}\ \xi_{i,(1)}\ \text{are independent}.
\end{align*}
Thus we have
\begin{align*}
    S_{n,p}^* =& \sum_{i=1}^n a_i \xi_i^\T\xi_i + \sum_{i=1}^{k-1}\sum_{k=2}^n b_{i,k} \xi_i^\T\xi_k\\
    =& \left(\sum_{i=1}^n a_i U_i^\T U_i + \sum_{i=1}^{k-1}\sum_{k=2}^n b_{i,k} U_i^\T U_k\right)
    + \sum_{i=1}^n a_i \xi_{i,(1)}^\T\xi_{i,(1)} + 2\sum_{i=1}^n a_i U_i^\T V_i + \sum_{i=1}^n a_i V_i^\T V_i\\
    &+\sum_{i=1}^{k-1}\sum_{k=2}^n b_{i,k} \xi_{i,(1)}^\T\xi_{k,(1)} + 2\sum_{i=1}^{k-1}\sum_{k=2}^n b_{i,k} U_i^\T V_k + \sum_{i=1}^{k-1}\sum_{k=2}^n b_{i,k} V_i^\T V_k\\
    :=& S^*_1+\Theta_1+\Theta_2+\Theta_3+\Theta_4+\Theta_5+\Theta_6 := S^*_1+S^*_2.
\end{align*}
We claim that, for any $\epsilon>0$, there exists a sequence of positive constants $t:=t_p>0$ with $t_p\to\infty$ such that
\begin{align}\label{claim:S2}
    \pr(|\Theta_i|\geq\epsilon \nu_{n,p})\leq p^{-t},\ i=1,\ldots,6,
\end{align}
for sufficiently large $p$. Consequently, $\pr(|S^*_2|/\nu_{n,p}\geq\epsilon)\leq p^{-t}$ for some $t\to\infty$ and sufficiently large $p$.

Write for short $\widetilde{S}^*_1=\{S^*_1-(n+2)p\}/\nu_{n,p}$ and thus $A_p(x)=\Big\{\widetilde{S}^*_1+S^*_2/\nu_{n,p}\leq x\Big\}$. By \eqref{indep:U-xi1}, we have
\begin{align*}
    \pr\big(A_p(x)B_{j_1}\cdots B_{j_d}\big)
    &\leq \pr\big(A_p(x)B_{j_1}\cdots B_{j_d}, |S^*_2|/\nu_{n,p}<\epsilon\big) + p^{-t}\\
    &\leq \pr(\widetilde{S}^*_1\leq x+\epsilon, B_{j_1}\cdots B_{j_d}) + p^{-t}\\
    &= \pr(\widetilde{S}^*_1\leq x+\epsilon)\cdot\pr(B_{j_1}\cdots B_{j_d}) + p^{-t}.
\end{align*}
We also have
\begin{align*}
    \pr(\widetilde{S}^*_1\leq x+\epsilon)
    &\leq \pr(\widetilde{S}^*_1\leq x+\epsilon, |S_2|/\nu_{n,p}<\epsilon) + p^{-t}\\
    &\leq \pr\big(A_p(x+2\epsilon)\big) + p^{-t}.
\end{align*}
Consequently,
\begin{align}\label{indep:A-B-leq}
    \pr\big(A_p(x)B_{j_1}\cdots B_{j_d}\big)
    \leq \pr\big(A_p(x+2\epsilon)\big)\cdot\pr(B_{j_1}\cdots B_{j_d}) + 2p^{-t}.
\end{align}

On the other hand, by the facts that
\begin{align*}
    \pr(\widetilde{S}^*_1\leq x-\epsilon)\cdot\pr(B_{j_1}\cdots B_{j_d})
    &= \pr(\widetilde{S}^*_1\leq x-\epsilon, B_{j_1}\cdots B_{j_d})\\
    &\leq \pr(\widetilde{S}^*_1\leq x-\epsilon, B_{j_1}\cdots B_{j_d}, |S_2|/\nu_{n,p}<\epsilon) + p^{-t}\\
    &\leq \pr\big(A_p(x)B_{j_1}\cdots B_{j_d}\big) + p^{-t}
\end{align*}
and
\begin{align*}
    \pr\big(A_p(x-2\epsilon)\big)
    &\leq \pr\big(A_p(x-2\epsilon), |S_2|/\nu_{n,p}<\epsilon\big) + p^{-t}\\
    &\leq \pr(\widetilde{S}^*_1\leq x-\epsilon) + p^{-t},
\end{align*}
we have
\begin{align}\label{indep:A-B-geq}
    \pr\big(A_p(x)B_{j_1}\cdots B_{j_d}\big)
    \geq \pr\big(A_p(x-2\epsilon)\big)\cdot\pr(B_{j_1}\cdots B_{j_d}) - 2p^{-t}.
\end{align}

By \eqref{indep:A-B-leq} and \eqref{indep:A-B-geq}, we conclude that
\begin{align*}
    \big|\pr\big(A_p(x)B_{j_1}\cdots B_{j_d}\big) - \pr\big(A_p(x)\big)\pr(B_{j_1}\cdots B_{j_d})\big|
    \leq \Delta_{p,\epsilon}\pr(B_{j_1}\cdots B_{j_d}) + 2p^{-t},
\end{align*}
for sufficiently large $p$, where
\begin{align*}
    \Delta_{p,\epsilon} &= |\pr\big(A_p(x)\big) - \pr\big(A_p(x+2\epsilon)\big)| + |\pr\big(A_p(x)\big) - \pr\big(A_p(x-2\epsilon)\big)|\\
    &= \pr\big(A_p(x+2\epsilon)\big) - \pr\big(A_p(x-2\epsilon)\big)
\end{align*}
since $\pr\big(A_p(x)\big)$ is increasing in $x$.
By running over all possible combinations of $1\leq j_1<\cdots<j_d\leq p$,
\begin{align*}
    \zeta(p,d)\leq \Delta_{p,\epsilon}\cdot H(p,d) + 2\binom{p}{d}p^{-t}.
\end{align*}

Since $\pr(A_p)\to \Phi(x)$, $\Delta_{p, \epsilon}\to\Phi(x+2\epsilon)-\Phi(x-2\epsilon)$, as $p\to\infty$. This implies that $\lim_{\epsilon\downarrow 0}\limsup_{p\to\infty}\Delta_{p, \epsilon}=0$.
Since for each $d\geq 1$, $H(d, p)\to \frac{1}{d!}\exp(-dx/2)$ as $p\to\infty$, we have $\limsup_{p\to\infty}H(p,d)<\infty$.
Due to the fact that $\binom{p}{d} \leq p^d$, for fixed $d\geq 1$, first sending $p\to\infty$ and then sending $\epsilon\downarrow 0$, we get $\lim_{p\to\infty}\zeta(p,d)=0$ for each $d\geq 1$. Hence Proposition \ref{lem:indep:A-B} follows.

It remains to prove that the claim \eqref{claim:S2} indeed holds.

\textit{Verification of \eqref{claim:S2}}: Due to the Gaussianity, $\nu_{n,p}\sim\sqrt{\frac{2\pi^2-18}{3}}n\sqrt{\tr(R^2)}$. Then,
\begin{align*}
    \pr(|\Theta_1|\geq\epsilon\nu_{n,p})
    &= \pr\left(\left|\sum_{i=1}^{n} a_{i}\xi_{i,(1)}^{\T}\xi_{i,(1)}\right|\geq \epsilon \nu_{n,p}\right)\\
    &\leq \pr\left(\max_{1\leq i\leq n}\left|\xi_{i,(1)}^{\T}\xi_{i,(1)}\right|\geq C_{\epsilon} \sqrt{\tr(R^2)}\right)\\
    &\leq n\pr\left(\left|\xi_{i,(1)}^{\T}\xi_{i,(1)}\right|\geq C_{\epsilon} \sqrt{\tr(R^2)}\right)
    \leq n\exp(-C_{\epsilon)d^{-1}p^{1/2}},
\end{align*}
where the last inequality follows by Lemma 7.7 in \cite{feng}, and $C_{\epsilon}$ denotes some positive constant depending on $\epsilon$. Similarly,
\begin{align*}
    \pr(|\Theta_2|\geq \epsilon\nu_{n,p})
    &= \pr\left(\left|2\sum_{i=1}^{n}a_{i}U_{i}^\T V_{i}\right|\geq \epsilon \nu_{n,p}\right)\\
    &\leq \pr\left(\max_{1\leq i\leq n}\left|U_{i}^\T V_{i}\right|\geq C_{\epsilon} \sqrt{\tr(R^2)}\right)\\
    &\leq n\pr\left(\left|U_{i}^\T V_{i}\right|\geq C_{\epsilon} \sqrt{\tr(R^2)}\right)
    \leq n\exp(-C_{\epsilon)\frac{\sqrt{2\tr(R^2)}}{\lambda_{\max}(R)}}
\end{align*}
and
\begin{align*}
    \pr(|\Theta_3|\geq \epsilon\nu_{n,p})
    &= \pr\left(\left|\sum_{i=1}^{n}a_{i}V_{i}^\T V_{i}\right|\geq \epsilon \nu_{n,p}\right)\\
    &\leq \pr\left(\max_{1\leq i\leq n}\left|V_{i}^\T V_{i}\right|\geq C_{\epsilon} \sqrt{\tr(R^2)}\right)\\
    &\leq n\pr\left(\left|V_{i}^\T V_{i}\right|\geq C_{\epsilon} \sqrt{\tr(R^2)}\right)
    \leq n\exp(-C_{\epsilon)\frac{\sqrt{2\tr(R^2)}}{\lambda_{\max}(R)}}.
\end{align*}
Additionally, define $\upsilon_d^2=\frac{2\pi^2-18}{3}n^2\tr(R_{11}^2)$ and we have
\begin{align*}
    \pr(|\Theta_4|\geq \epsilon\nu_{n,p})
    &= \pr\left(\left|\sum_{i=1}^{k-1}\sum_{k=2}^{n}b_{i,k} \xi_{i,(1)}^{\T}\xi_{k,(1)}\right|\geq \epsilon \nu_{n,p}\right)\\
    &= \pr\left(\left|\upsilon_d^{-1}\sum_{i=1}^{k-1}\sum_{k=2}^{n}b_{i,k} \xi_{i,(1)}^{\T}\xi_{k,(1)}\right|\geq \epsilon \sqrt{\frac{\tr(R^2)}{\tr(R_{11}^2)}}\right)\\
    &\leq \exp\left(-\epsilon \sqrt{\frac{\tr(R^2)}{\tr(R_{11}^2)}}\right)\E\left(\exp\left(\left|\upsilon_d^{-1}\sum_{i=1}^{k-1}\sum_{k=2}^{n}b_{i,k} \xi_{i,(1)}^{\T}\xi_{k,(1)}\right|\right)\right)\\
    &\leq \log n\exp\left(-\epsilon \sqrt{\frac{\tr(R^2)}{\tr(R_{11}^2)}}\right),
\end{align*}
where the last inequality follows since
\begin{align*}
    \frac{\upsilon_d^{-1}\sum_{i=1}^{k-1}\sum_{k=2}^{n}b_{i,k} \xi_{i,(1)}^{\T}\xi_{k,(1)}}{\log \log n}\to 0\ \text{a.s.},
\end{align*}
by the law of the iterated logarithm of zero-mean square integrable martingale (see Theorem 4.8 in \cite{MR624435}).
Similarly,
\begin{align*}
    \pr(|\Theta_5|\geq \epsilon\nu_{n,p})
    &= \pr\left(\left|2\sum_{i=1}^{k-1}\sum_{k=2}^{n}b_{i,k} U_{i}^{\T}V_{k}\right|\geq \epsilon \nu_{n,p}\right)\\
    &\leq \log n\exp\left(-\frac{\epsilon}{2} \sqrt{\frac{\tr(R^2)}{\tr(R_{22\cdot 1}R_{21}R_{11}^{-1}R_{12})}}\right)\ \text{and}\ \\
    \pr(|\Theta_6|\geq \epsilon\nu_{n,p})
    &= \pr\left(\left|2\sum_{i=1}^{k-1}\sum_{k=2}^{n}b_{i,k} U_{i}^{\T}V_{k}\right|\geq \epsilon \nu_{n,p}\right)\\
    &\leq \log n\exp\left(-\frac{\epsilon}{2} \sqrt{\frac{\tr(R^2)}{\tr((R_{21}R_{11}^{-1}R_{12}))^2}}\right).
\end{align*}
It is then easy to see that \eqref{claim:S2} holds.
\end{proof}

\subsection{Sub-Gaussian data sequences}

By Assumption \ref{asmp:corr_spatial-ev}, it can be further verified that
\begin{align*}
    S_{n,p}-(n+2)p &= \sum_{i=1}^{n-1}\sum_{k=2}^{n} b_{i,k}\xi_i^{\T}\xi_k + o_p(\nu^2_{n,p}),
\end{align*}
where $\nu^2_{n,p}=\frac{2\pi^2-18}{3}n^2\tr(R^2)\{1+o(1)\}:=\sigma_S^2\{1+o(1)\}$.
Define
\begin{align*}
    W(x_1,\ldots,x_n) = \sigma_S^{-1}\sum_{i=1}^{n-1}\sum_{k=2}^{n}b_{i,k} x_i^{\T}x_k.
\end{align*}
According to Lemma \ref{lem:asy_indep}, it suffices to show that
\begin{align}\label{thm*:indep}
    \pr\left(W(\xi_1,\ldots,\xi_n)\leq x, \max_{1\leq j\leq p} \max_{1\leq k\leq n}|C_{0,j}(k)|/\sigma_j\leq u_p\{\exp(-y)\}\right)\to\Phi(x)\cdot\exp\{-\exp(-y)\}.
\end{align}

For $z=(z_{1},\ldots,z_{q})^\T\in\mathbb{R}^{q}$, we consider a smooth approximation of the maximum function $z\mapsto \max_{1\leq j\leq q}z_j$, namely,
$$
    F_{\beta}(z) := \beta^{-1}\log\left(\sum_{j=1}^{q}\exp(\beta z_{j})\right),
$$
where $\beta>0$ is the smoothing parameter that controls the level of approximation.
An elementary calculation shows that for all $z\in\mathbb{R}^{q}$,

$$
    0 \leq F_{\beta}(z)-\max_{1\leq j\leq q}z_{j} \leq \beta^{-1} \log q,
$$
see, for example, \cite{MR3161448}.
Without loss of generality, we assume that $\sigma_j=1$ for $j=1,\ldots,p$. Define
\begin{align*}
    V(x_1,\ldots,x_n)=\beta^{-1}\log \left(\sum_{j=1}^p\sum_{k=1}^n \exp\left(\beta n^{-1/2}\sum_{t=1}^{n}x_{tj}\delta_{tk}\right)\right)
\end{align*}
where $\delta_{tk}=I(t\le k)-k/n$.
By setting $\beta=n^{1/8}\log (np)$, \eqref{thm*:indep} is equivalent to
\begin{align}\label{thm**:indep}
    \pr\left(W(\xi_1,\ldots,\xi_n)\leq x, V(\xi_1,\ldots,\xi_n)\leq u_p\{\exp(-y)\}\right)\to\Phi(x)\cdot\exp\{-\exp(-y)\}.
\end{align}
Suppose $\{Y_1,\ldots,Y_n\}$ are i.i.d. from $N(0,\Sigma)$ and are independent of $\{\xi_1,\ldots,\xi_n\}$. The key idea is to show that $(W(\xi_1,\ldots,\xi_n),V(\xi_1,\ldots,\xi_n))$ has the same limit distribution as $(W(Y_1,\ldots,Y_n),V(Y_1,\ldots,Y_n))$, form which \eqref{thm**:indep} holds according to the result in Section \ref{indep:Gaussian}.

Let $\mathscr{C}_{b}^{3}(\mathbb{R})$ denote the class of bounded functions with bounded and continuous derivatives up to order 3. It is known that a sequence of random variables $\{Z_{n}\}_{n=1}^{\infty}$ converges weakly to a random variable $Z$ if and only if for every $f\in\mathscr{C}_{b}^{3}(\mathbb{R}), \E(f(Z_n))\to\E(f(Z))$; see, e.g., \cite{MR762984}.
It suffices to show that
\begin{align*}
    \E\left\{f(W(\xi_1,\ldots,\xi_n),V(\xi_1,\ldots,\xi_n))\right\}-\E\left\{f(W(Y_1,\ldots,Y_n),V(Y_1,\ldots,Y_n))\right\} \to 0,
\end{align*}
for every $f\in \mathscr{C}_{b}^{3}(\mathbb{R}^2)$ as $(n,p)\to\infty$.
We introduce $W_d=W(\xi_1,\ldots,\xi_{d-1},Y_d,\ldots,Y_n)$ and $V_d=V(\xi_1,\ldots,\xi_{d-1},Y_d,\ldots,Y_n)$ for $d=1,\ldots,n+1$. Then
\begin{align*}
    |\E\{&f(W(\xi_1,\ldots,\xi_n),V(\xi_1,\ldots,\xi_n))\}-\E\{f(W(Y_1,\ldots,Y_n),V(Y_1,\ldots,Y_n))\}|\\
    &\leq \sum_{d=1}^n\left|\E\{f(W_d,V_d)\}-\E\{f(W_{d+1},V_{d+1})\}\right|.
\end{align*}

Let
\begin{align*}
    W_{d,0} &= \sigma_S^{-1}\left(\sum_{i=1}^{k-1}\sum_{k=2}^{d-1} b_{i,k} \xi_i^{\T}\xi_k + \sum_{i=d+1}^{k-1}\sum_{k=d+2}^{n} b_{i,k} Y_{i}^{\T}Y_{k} + \sum_{i=1}^{d-1}\sum_{k=d+1}^{n} b_{i,k} \xi_i^{\T}Y_{k}\right)\ \text{and}\ \\
    V_{d,0} &= \beta^{-1}\log\left(\sum_{j=1}^p\sum_{k=1}^n\exp\left(\beta n^{-1/2}\sum_{t=1}^{d-1}\xi_{tj}\delta_{tk}+\beta n^{-1/2}\sum_{t=d+1}^{n}Y_{tj}\delta_{tk}\right)\right),
\end{align*}
which only rely on $\mathcal{F}_d=\sigma\{\xi_1,\ldots,\xi_{d-1},Y_{d+1},\ldots,Y_n\}$.
By Taylor's expansion, we have
\begin{align*}
    f(W_d,V_d) - f(W_{d,0},V_{d,0}) =& f_1(W_{d,0},V_{d,0})(W_d-W_{d,0}) + f_2(W_{d,0},V_{d,0})(V_d-V_{d,0})\\
    &+ \frac{1}{2}f_{11}(W_{d,0},V_{d,0})(W_d-W_{d,0})^2 + \frac{1}{2}f_{22}(W_{d,0},V_{d,0})(V_d-V_{d,0})^2\\
    &+ \frac{1}{2}f_{12}(W_{d,0},V_{d,0})(W_d-W_{d,0})(V_d-V_{d,0})\\
    &+ O(|(V_d-V_{d,0})|^3)+O(|(W_d-W_{d,0})|^3)
\end{align*}
and
\begin{align*}
    f(W_{d+1},V_{d+1}) - f(W_{d,0},V_{d,0}) =& f_1(W_{d,0},V_{d,0})(W_{d+1}-W_{d,0})+f_2(W_{d,0},V_{d,0})(V_{d+1}-V_{d,0})\\
    &+ \frac{1}{2}f_{11}(W_{d,0},V_{d,0})(W_{d+1}-W_{d,0})^2 + \frac{1}{2}f_{22}(W_{d,0},V_{d,0})(V_{d+1}-V_{d,0})^2\\
    &+ \frac{1}{2}f_{12}(W_{d,0},V_{d,0})(W_{d+1}-W_{d,0})(V_{d+1}-V_{d,0})\\
    &+ O(|(V_{d+1}-V_{d,0})|^3)+O(|(W_{d+1}-W_{d,0})|^3),
\end{align*}
where for $f:=f(x,y)$, $f_1(x,y)=\frac{\partial f}{\partial x}$, $f_2(x,y)=\frac{\partial f}{\partial y}$, $f_{11}(x,y)=\frac{\partial f^2}{\partial^2 x}$, $f_{22}(x,y)=\frac{\partial f^2}{\partial^2 y}$ and $f_{12}(x,y)=\frac{\partial f^2}{\partial x \partial y}$.

Notice that
\begin{align*}
    W_{d}-W_{d,0} = \sum_{i=1}^{d-1} b_{i,d} \xi_i^{\T} Y_d + \sum_{k=d+1}^n b_{d,k} Y_d^{\T} Y_{k}\ \text{and}\
    W_{d+1}-W_{d,0} = \sum_{i=1}^{d-1} b_{i,d} \xi_i^{\T} \xi_d + \sum_{k=d+1}^n b_{d,k} \xi_d^{\T} Y_{k}.
\end{align*}
Due to $\E(\xi_t)=\E(Y_t)=0$ and $\E(\xi_t\xi_t^{\T})=\E(Y_tY_t^{\T})$, it can be verified that
\begin{align*}
    \E(W_{d}-W_{d,0}\mid\mathcal{F}_d)=\E(W_{d+1}-W_{d,0}\mid\mathcal{F}_d)\ \text{and}\ \E((W_{d}-W_{d,0})^2\mid\mathcal{F}_d)=\E((W_{d+1}-W_{d,0})^2\mid\mathcal{F}_d).
\end{align*}
Hence,
\begin{align*}
    \E\{f_1(W_{d,0},V_{d,0})(W_{d}-W_{d,0})\} &= \E\{f_1(W_{d,0},V_{d,0})(W_{d+1}-W_{d,0})\}\ \text{and}\ \\
    \E\{f_{11}(W_{d,0},V_{d,0})(W_{d}-W_{d,0})^2\} &= \E\{f_{11}(W_{d,0},V_{d,0})(W_{d+1}-W_{d,0})^2\}.
\end{align*}
Consider $V_{d}-V_{d,0}$.
For $l=k+(j-1)n$, let $z_{d,0,l}=n^{-1/2}\sum_{t=1}^{d-1}\xi_{tj}\delta_{tk}+ n^{-1/2}\sum_{t=d+1}^{n}Y_{tj}\delta_{tk}$, $z_{d,l}=z_{d,0,l}+n^{-1/2}Y_{dj}\delta_{dk}$ and $z_{d+1,l}=z_{d,0,l}+n^{-1/2}\xi_{dj}\delta_{dk}$.
Define $\bm z_{d,0}=(z_{d,0,1},\ldots,z_{d,0,np})^{\T}$ and $\bm z_{d}=(z_{d,1},\ldots,z_{d,np})^{\T}$.
By Taylor's expansion, we have
\begin{align}\label{vd}
    V_d-V_{d,0} =& \sum_{l=1}^{np}\partial_l F_\beta(\bm z_{d,0})(z_{d,l}-z_{d,0,l}) + \frac{1}{2}\sum_{l=1}^{np}\sum_{k=1}^{np}\partial_k\partial_l F_\beta(\bm z_{d,0})(z_{d,l}-z_{d,0,l})(z_{d,k}-z_{d,0,k})\nonumber\\
    &+ \frac{1}{6}\sum_{l=1}^{np}\sum_{k=1}^{np}\sum_{v=1}^{np}\partial_v\partial_k\partial_l F_\beta(\bm z_{d,0}+\delta(\bm z_d-\bm z_{d,0}))(z_{d,l}-z_{d,0,l})(z_{d,k}-z_{d,0,k})(z_{d,v}-z_{d,0,v})
\end{align}
for some $\delta\in(0,1)$.
Again, due to $\E(\xi_t)=\E(Y_t)=0$ and $\E(\xi_t\xi_t^{\T})=\E(Y_tY_t^{\T})$, we can verify that
\begin{align*}
    \E\{(z_{d,l}-z_{d,0,l})\mid\mathcal{F}_d\} = \E\{(z_{d+1,l}-z_{d,0,l})\mid\mathcal{F}_d\}\ \text{and}\
    \E\{(z_{d,l}-z_{d,0,l})^2\mid\mathcal{F}_d\} = \E\{(z_{d+1,l}-z_{d,0,l})^2\mid\mathcal{F}_d).
\end{align*}
By Lemma A.2 in \cite{MR3161448}, we have
\begin{align*}
    \left|\sum_{l=1}^{np}\sum_{k=1}^{np}\sum_{v=1}^{np}\partial_v\partial_k\partial_l F_\beta(\bm z_{d,0}+\delta(\bm z_d-\bm z_{d,0}))\right|
    \leq C\beta^2
\end{align*}
for some positive constant $C$.
By Assumption \ref{asmp:noise},
$\pr\left(\max_{1\leq t\leq n,1\leq j\leq p} |\xi_{tj}| > C\log(np)\right)\to 0$,
and since $Y_{tj}\sim N(0,1)$, $\pr\left(\max_{1\leq t\leq n,1\leq j\leq p}|Y_{tj}| > C\log (np)\right)\to 0$.
Hence,
\begin{align*}
    \Big|\frac{1}{6}&\sum_{l=1}^{np}\sum_{k=1}^{np}\sum_{v=1}^{np}\partial_v\partial_k\partial_l F_\beta(\bm z_{d,0}+\delta(\bm z_d-\bm z_{d,0}))(z_{d,l}-z_{d,0,l})(z_{d,k}-z_{d,0,k})(z_{d,v}-z_{d,0,v})\Big|\\
    &\leq C\beta^2 n^{-3/2}\log^3(np)
\end{align*}
holds with probability approaching one.
Consequently, we have, with probability approaching one,
\begin{align*}
    \left|\E\{f_2(W_{d,0},V_{d,0})(V_d-V_{d,0})\}-\E\{f_2(W_{d,0},V_{d,0})(V_{d+1}-V_{d,0})\}\right|\leq C\beta^2 n^{-3/2}\log^3(np).
\end{align*}
Similarly, it can be verified that
\begin{align*}
    \left|\E\{f_{22}(W_{d,0},V_{d,0})(V_d-V_{d,0})^2\} - \E\{f_{22}(W_{d,0},V_{d,0})(V_{d+1}-V_{d,0})^2\}\right|
    &\leq C\beta^2 n^{-3/2}\log^3(np)
\end{align*}
and
\begin{align*}
    \Big|\E\{&f_{12}(W_{d,0},V_{d,0})(W_{d}-W_{d,0})(V_{d}-V_{d,0})\} - \E\{f_{12}(W_{d,0},V_{d,0})(W_{d+1}-W_{d,0})(V_{d+1}-V_{d,0})\}\Big|\\
    &\leq C\beta^2 n^{-3/2}\log^3(np).
\end{align*}
By \eqref{vd}, $\E(|V_{d}-V_{d,0}|^3)=O(n^{-3/2}\log^3(np))$.
According to the proof of Theorem 1 in \cite{10.1142/s201032631950014x}, $\E((W_{d}-W_{d,0})^4)=O(n^{-2})$ and thus
\[
\sum_{d=1}^n\E|W_{d}-W_{d,0}|^3\leq\sum_{d=1}^n\{\E(W_{d}-W_{d,0})^4)\}^{3/4}\leq C' n^{-1/2}
\]
for some positive constant $C'$,
Combining all facts together, we conclude that
\begin{align*}
    \sum_{d=1}^n\left|\E\{f(W_d,V_d)\}-\E\{f(W_{d+1},V_{d+1})\}\right|\leq C\beta^2 n^{-1/2}\log^3(np) + C'n^{-1/2}\to 0,
\end{align*}
as $(n,p)\to\infty$. The conclusion follows.

\section{Proof of Theorem \ref{alter:DMS}}
It suffices to show the conclusion holds for Gaussian data sequences. According to the proof of Theorem 2 in \cite{10.1142/s201032631950014x}, we have
\begin{align*}
    S_{n,p} &= \Delta_{S}+\sum_{i=1}^{n-1}\sum_{k=2}^n b_{i,k}\xi_i^\T\xi_k+o_p\Big(\sqrt{\var(S_{n,p})}\Big)\\
    &= \Delta_{S}+\sum_{i=1}^{n-1}\sum_{k=2}^n\sum_{j\in \mathcal{A}} b_{i,k}\xi_{ij}\xi_{kj}+\sum_{i=1}^{n-1}\sum_{k=2}^n\sum_{j\in \mathcal{A}^c} b_{i,k}\xi_{ij}\xi_{kj}+o_p\Big(\sqrt{\var(S_{n,p})}\Big)\\
    &= \Delta_{S}+\sum_{i=1}^{n-1}\sum_{k=2}^n\sum_{j\in \mathcal{A}^c} b_{i,k}\xi_{ij}\xi_{kj}+o_p\Big(\sqrt{\var(S_{n,p})}\Big),
\end{align*}
since
\begin{align*}
    \var\left(\sum_{i=1}^{n-1}\sum_{k=2}^n\sum_{j\in \mathcal{A}} b_{i,k}\xi_{ij}\xi_{kj}\right)
    =\frac{2\pi^2-18}{3}n^2\tr(R_{\mathcal{A}}^2) \{1+o(1)\}=o\{\var(S_{n,p})\}
\end{align*}
and $\tr(R_{\mathcal{A}}^2)=O(|\mathcal{A}|)=o(\tr(R^2))$, where $R_{\mathcal{A}}$ is the sub-matrix of $R$ with rows and columns in $\mathcal{A}$.

Then we rewrite
\begin{align*}
    M_{n,p}=\max\{\max_{j\in \mathcal{A}}\max_{1\le i\le n-1}|C_{0,j}(k)|, \max_{j\in \mathcal{A}^c}\max_{1\le i\le n-1}|C_{0,j}(k)|\}.
\end{align*}
According to Theorem \ref{null:DMS}, we have known that $\sum_{i=1}^{n-1}\sum_{k=2}^n\sum_{j\in \mathcal{A}^c} b_{i,k}\xi_{ij}\xi_{kj}$ is asymptotically independent of $\max_{j\in \mathcal{A}^c}\max_{1\le i\le n-1}|C_{0,j}(k)|$. Hence it suffices to show that $\sum_{i=1}^{n-1}\sum_{k=2}^n\sum_{j\in \mathcal{A}^c} b_{i,k}\xi_{ij}\xi_{kj}$ is asymptotically independent of $\xi_{ij}, j\in \mathcal{A}$.

Without loss of generality, we assume $\mathcal{A}=\{j_1,\cdots,j_d\}$. For each $i=1,\ldots,n$, let $\xi_{i,(1)}=(\xi_{i,j_1},\ldots,\xi_{i,j_d})^\T$ and $\xi_{i,(2)}=(\xi_{i,j_{d+1}},\ldots,\xi_{i,j_p})^\T$, and $R_{kl}=\cov(\xi_{i,(k)}, \xi_{i,(l)})$ for $k,l\in\{1,2\}$. By Lemma \ref{lem:X2-rX1}, $\xi_{i,(2)}$ can be decomposed as $\xi_{i,(2)}=U_i+V_i$, where $U_i:=\xi_{i,(2)}-R_{21}R_{11}^{-1}\xi_{i,(1)}$ and $V_i:=R_{21}R_{11}^{-1}\xi_{i,(1)}$ satisfying that $U_i\sim N(0, R_{22}-R_{21}R_{11}^{-1}R_{12})$, $V_i\sim N(0, R_{21}R_{11}^{-1}R_{12})$ and
\begin{align}\label{indep:U-xi1}
    U_i\ \text{and}\ \xi_{i,(1)}\ \text{are independent}.
\end{align}
We have
\begin{align*}
    \sum_{i=1}^{n-1}\sum_{k=2}^n\sum_{j\in \mathcal{A}^c} b_{i,k}\xi_{ij}\xi_{kj}
    =\sum_{i=1}^{k-1}\sum_{k=2}^n b_{i,k} \xi_{i,(1)}^\T\xi_{k,(1)} + 2\sum_{i=1}^{k-1}\sum_{k=2}^n b_{i,k} U_i^\T V_k + \sum_{i=1}^{k-1}\sum_{k=2}^n b_{i,k} V_i^\T V_k.
\end{align*}
By using arguments similar to those in the proof of Lemma \ref{lem:indep:A-B}, we have
\begin{align*}
    &\pr\left(\sum_{i=1}^{k-1}\sum_{k=2}^n b_{i,k} U_i^\T V_k\geq \epsilon \sqrt{\var(S_{n,p})}\right)\le \log n \exp(-c_{\epsilon}p^{1/2}/d^{1/2})\to 0\ \text{and}\ \\
    &\pr\left(\sum_{i=1}^{k-1}\sum_{k=2}^n b_{i,k} V_i^\T V_k\geq \epsilon \sqrt{\var(S_{n,p})}\right)\le \log n \exp(-c_{\epsilon}p^{1/2}/d^{1/2})\to 0,
\end{align*}
since $d=|\mathcal{A}|=o(p/(\log \log p)^2)$ and $p\lesssim n^\nu$. Consequently, we conclude that
\[
    \sum_{i=1}^{n-1}\sum_{k=2}^n\sum_{j\in \mathcal{A}^c} b_{i,k}\xi_{ij}\xi_{kj}=\sum_{i=1}^{k-1}\sum_{k=2}^n b_{i,k} \xi_{i,(1)}^\T\xi_{k,(1)}+o_p\left(\sqrt{\var(S_{n,p})}\right).
\]
By Lemma \ref{lem:asy_indep} and \eqref{indep:U-xi1}, we have $\sum_{i=1}^{n-1}\sum_{k=2}^n\sum_{j\in \mathcal{A}^c} b_{i,k}\xi_{ij}\xi_{kj}$ is asymptotically independent of $\xi_{i,(1)}$. Hence Theorem \ref{alter:DMS}--(i) follows. The proof of \ref{alter:DMS}--(ii) is similar, and thus is omitted.

\section{Some useful facts}

Let $\mathcal{W}$ be a standard Brownian motion and $\mathcal{V}$ be an Ornstein-Uhlenbeck process with $\E\big(\mathcal{V}(t)\big)=0$ and $\E\big(\mathcal{V}(t)\mathcal{V}(s)\big)=\exp(-|t-s|/2)$.

\begin{lemma}\label{lem:sup_Wt}
For any $x>0$, $\pr\left(\sup_{0\leq t\leq 1}\left|\mathcal{W}_{t}\right|\geq x\right)\leq\frac{4}{\sqrt{2\pi}x}\exp(-x^{2}/2)$.
\end{lemma}
\begin{proof}
See, for example, \cite{MR1121940} or (B.16) in \cite{10.1214/15-aos1347}.
\end{proof}

\begin{lemma}\label{lem:sup_Bt}
For any $x>0$, $\pr(\sup_{0\leq t\leq 1}|\mathcal{W}_t-t\mathcal{W}_1|\geq x)=2\sum_{k=1}^{\infty}(-1)^{k+1}\exp(-2k^2x^2)$.
\end{lemma}
\begin{proof}
See, for example, \cite{MR838963} or Lemma B.4 in \cite{10.1214/15-aos1347}.
\end{proof}

\begin{lemma}\label{lem:V(t)}
For all $T>0$, $\pr\big(\sup_{0\leq t\leq T}|\mathcal{V}(t)|>x\big) = \frac{x\exp(-x^2/2)}{\sqrt{2\pi}}\{T-x^{-2}T+4x^{-2}+O(x^{-4})\}$, as $x\to\infty$.
\end{lemma}
\begin{proof}
See Theorem A.3.3 in \cite{csorgo1997limit}.
\end{proof}

\begin{lemma}\label{lem:X2-rX1}
Let $X\sim N(\mu,\Sigma)$ with invertible $\Sigma$, and partition $X$, $\mu$ and $\Sigma$ as
\begin{eqnarray*}
X=
\begin{pmatrix}
X_1\\
X_2
\end{pmatrix}
,\
\mu=
\begin{pmatrix}
\mu_1\\
\mu_2
\end{pmatrix}
\ \text{and}\
\Sigma=
\begin{pmatrix}
\Sigma_{11} & \Sigma_{12}\\
\Sigma_{21} & \Sigma_{22}
\end{pmatrix}.
\end{eqnarray*}
Then $X_2-\Sigma_{21}\Sigma_{11}^{-1}X_1\sim N(\mu_2-\Sigma_{21}\Sigma_{11}^{-1}\mu_1, \Sigma_{22\cdot 1})$ and is independent of $X_1$, where $\Sigma_{22\cdot 1}=\Sigma_{22}-\Sigma_{21}\Sigma_{11}^{-1}\Sigma_{12}$.
\end{lemma}
\begin{proof}
See Theorem 1.2.11 of \cite{muirhead2009aspects}.
\end{proof}

\begin{lemma}\label{lem:Feng}
For a graph $G$, we say vertices $i$ and $j$ are neighbors if there is an edge between them. For a set $A$, we write $|A|$ for its cardinality.
Let $G=(V, E)$ be an undirected graph with $n=|V|\geq 4$ vertices. Write $V=\{v_1, \ldots, v_n\}$. Assume each vertex in $V$ has at most $q$ neighbors. Let $G_t$ be the set of subgraphs of $V$ such that each subgraph has $t$ vertices and  at least one edge. The following are true.

(i) $|G_t| \leq qn^{t-1}$ for any $2\leq t \leq  n$.

(ii) Fix integer $t$ with $2\leq t \leq  n$. Let $G_t' \subset G_t$ such that each member of $G_t'$ is a clique, that is, any two vertices are neighbors. Then $|G_t'|\leq nq^{t-1}.$

The following conclusions are true for integer $t$ with $3\leq t \leq  n$.

(iii) For $j=2, \ldots, t-1$, let $H_j$ be the subset of $(i_1, \ldots, i_t)$ from $G_t$ satisfying the following: there exists a subgraph  $S$ of $\{i_1, \ldots, i_t\}$ with $|S|=j$ and without any edge such that any vertex from $\{i_1, \ldots, i_t\}\backslash S$ has at least two neighbors in $S$. Then $|H_j| \leq (qt)^{t-j+1}n^{j-1}$.

(iv) For $j=2, \ldots, t-1$, let $H_j'$ be the subset of $(i_1, \ldots, i_t)$ from $G_t$ satisfying the following: for any subgraph  $S$ of $\{i_1, \ldots, i_t\}$ with $|S|=j$ and without any edge, we know any vertex from $\{i_1, \ldots, i_t\}\backslash S$ has at least one neighbor in $S$. Then $|H_j'|\leq (qt)^{t-j}n^j.$
\end{lemma}
\begin{proof}
See Lemma 7.1 of \cite{feng}.
\end{proof}

\begin{lemma}\label{lem:asy_indep}
Let $\{(U, U_{p},\widetilde{U}_p)\in \mathbb{R}^3;\, p\geq 1\}$ and  $\{(V, V_{p},\widetilde{V}_p)\in \mathbb{R}^3;\, p\geq 1\}$ be two sequences of random variables with $U_p\to U$ and $V_p\to V$ in distribution as $p\to\infty.$ Assume $U$ and $V$ are continuous random variables and that
\begin{align*}
\widetilde{U}_p=U_p+o_p(1)\ \ \ \mbox{and}\ \ \ \widetilde{V}_p=V_p+o_p(1).
\end{align*}
If $U_p$ and $V_p$ are asymptotically independent, then $\widetilde{U}_p$ and $\widetilde V_p$ are also asymptotically independent.
\end{lemma}
\begin{proof}
See Lemma 7.10 of \cite{feng}.
\end{proof}

\end{document}